\documentclass[11pt,a4paper]{article}
\usepackage{style}

\title{Approximation Algorithms for \(\ell_p\)-Shortest Path and \(\ell_p\)-Group Steiner Tree}
\author{Yury Makarychev%
\thanks{Toyota Technological Institute at Chicago. Email: \texttt{yury@ttic.edu}. Supported by NSF Awards CCF-1955173, CCF-1934843, and ECCS-2216899.}
\and Max Ovsiankin%
\thanks{Toyota Technological Institute at Chicago. Email: \texttt{maxov@ttic.edu}. Partially supported by the Institute for Data, Econometrics, Algorithms, and Learning (IDEAL) with NSF Grant ECCS-2216899.}
\and Erasmo Tani%
\thanks{University of Chicago. Email: \texttt{etani@uchicago.edu}. Partially supported by the Institute for Data, Econometrics, Algorithms, and Learning (IDEAL) with NSF Grant ECCS-2216912.}}

\date{\today}
\begin{document}

\maketitle
\begin{abstract}
We present polylogarithmic approximation algorithms for variants of the Shortest Path, Group Steiner Tree, and Group ATSP problems with vector costs. In these problems, each edge $e$ has a vector cost $c_e\in\bbR_{\geq 0}^\ell$. For a feasible solution -- a path, subtree, or tour (respectively) -- we find the total vector cost of all the edges in the solution and then compute the $\ell_p$-norm of the obtained cost vector (we assume that $p\geq 1$ is an integer). Our algorithms for series-parallel graphs run in polynomial time and those for arbitrary graphs run in quasi-polynomial time.

To obtain our results, we introduce and use new flow-based Sum-of-Squares relaxations. 
We also obtain a number of hardness results.
\end{abstract}

\newpage
\tableofcontents
\newpage

\section{Introduction}
In this work, we study robust versions of network design problems. In the $\ell_p$-Shortest Path problem, we are given \(p \ge 1\), a graph $G=(V,E)$ with vector-valued edge costs  $c_e\in  \R_{\geq 0}^{\ell}$, and two vertices $s$ and $t$; the goal is to find a path $P$ from $s$ to $t$ in $G$ that minimizes the following cost:
\[
    \cost_{\ell_p}(P) = \norm{\sum_{e\in P}c_e}_p
\]

This problem is a natural generalization of the classical shortest path problem, but surprisingly has not received much attention till recently. 
The problem has been studied for $p=\infty$ under the name Robust Shortest Path. 
Aissi, Bazgan and Vanderpooten~\cite{aissi2007approximation} used dynamic programming to obtain a fully polynomial-time approximation scheme for the case when the number of coordinates $\ell$ is a constant and $p=\infty$ (this result generalizes to other $p$). Kasperski and Zieli\'nski~\cite{kasperski2009approximability} proved that $\ell_\infty$-Shortest Path is hard to approximate within $\log^{1-\eps} \ell$ for all \(\eps > 0\) unless $\mathrm{NP} \subseteq \mathrm{DTIME}(n^{\operatorname{polylog} n})$.
More recently, the same authors~\cite{kasperski2018approximating} gave an $O(\sqrt{n \log \ell / \log\log \ell})$-approximation to the problem by rounding a flow-based linear programming relaxation and proved that their LP has the integrality gap of $\Omega(\sqrt{n})$. In a recent breakthrough, Li, Xu, and Zhang \cite{li2023polylogarithmic} gave an $O(\log n \log \ell)$-approximation algorithm for $\ell_{\infty}$-Shortest Path with running time quasi-polynomial in the size of the input instance. In particular, their algorithm is the first polylogarithmic approximation known to date. They also show that the same approximation guarantees can be obtained in polynomial time for graphs of bounded treewidth, and they give a polynomial time $O(d \log \ell)$ approximation algorithm for series-parallel graphs, where $d$ is the depth (order) of the series-parallel decomposition of the input graph. No results were known for any other exponents $p\in(1,\infty)$. However, it is trivial to get an $\ell^{1-1/p}$-approximation by solving standard shortest path with edge costs $\|c_e\|_1$.

Introducing vector-valued costs to the graph's edges allows this model to capture a number of different applications. First, it allows us to describe a situation in which different parties, each corresponding to a different coordinate of the cost vectors, incur different cost when an edge is added to the solution. 
In this interpretation, as $p \to \infty$, the problem will increasingly favor paths in which every party simultaneously incurs small cost. Alternatively, each coordinate of the cost vectors may represent the cost incurred in terms of a different resource. This model would then allow one to balance minimizing the total amount of resources spent and ensuring that no single resource is depleted. Furthermore, one can think of this problem as providing an avenue for modeling robustness of a solution in the presence of uncertainty. Each coordinate would then represent the cost incurred by adding an edge in a distinct possible scenario, and the value of the $\ell_p$-Shortest Path problem would amount to a trade-off between average and worst-case cost among all scenarios. Finally, this problem generalizes congestion minimization in directed graphs (a fact that we prove in \Cref{sec:hardness-from-congestion-minimization}).

\paragraph{Our results for the $\ell_p$-Shortest Path problem}
In this paper, we introduce a natural flow-based sum-of-squares (SoS) relaxation for $\ell_p$-Shortest Path (\Cref{sec:sos-relaxation}) and present approximation algorithms for all integer $p\geq 1$. 

First, we give an $O(pd^{1-1/p})$-approximation algorithm for the problem running in $n^{O(p)}$ time for series-parallel graphs of depth/order $d$ (\Cref{sec:rounding}). We do this by considering a natural rounding algorithm for the SoS relaxation.
We prove the following theorem:
\begin{theorem}\label{thm:intro-series-parallel-guarantees}
There exists an approximation algorithm for the $\ell_p$-Shortest Path problem in series-parallel graphs that, given a series-parallel graph $G$ of order/depth $d$ and parameters $p\in \bbZ_{\geq 1}$ and $\varepsilon \in(0, 1)$,  finds a
$(1+\varepsilon)\bell_d(p)^{1/p}= O(pd^{1-1/p})$  approximation in time $m^{O(p)}/\varepsilon^{O(1)}$ (which is polynomial time when $p$ and $\varepsilon$ are fixed). Here, $\bell_d(p)$ is the $p$-{th} $d$-dimensional Bell number.\footnote{We provide a review of Bell numbers in the preliminaries.}
\end{theorem}
For graphs of small series-parallel order/depth $d\leq \log^* p$, the approximation factor is $\bell_d^{1/p} \leq O(p/\log^{(d)}p)$.  
Remarkably, in a complementary analysis given in \Cref{sec:tightness-of-analysis}, we show that the approximation factor $\bell_d(p)^{1/p}$ is tight for our rounding scheme. 
In all algorithms, we assume that $\ell$ is at most polynomial in $n$ (if not, the running times will also depend on $\ell$).

Then, we give a $O(p\log^{1-1/d} n)$-approximation for arbitrary graphs (\Cref{sec:general}), obtaining the following theorem:
\begin{theorem}\label{thm:intro-general-graphs-guarantees}
    There exists an approximation algorithm for the $\ell_p$-Shortest Path problem in arbitrary graphs that, given a graph $G$ and parameters $p\in\bbZ_{\geq 1}$ and $c\in(0,1/2)$, finds a $cp\log^{1-1/p}n$ approximation in time $m^{c e^{O(1/c)} \log n} = m^{O_c(\log n)}$. 
\end{theorem}
Note that when $p = \log \ell$, this result yields an $O(c\log n \log \ell)$ approximation. This is very similar to the approximation guarantee of $O(\log n \log k)$ by Li, Xu, and Zhang, except that in our algorithm $c$ can be an arbitrarily small constant.

\begin{remark} As previously discussed, an \(\ell^{1-1/p}\)-approximation is trivial to achieve in polynomial time by solving standard shortest path with edge costs \(\|c_e\|_1\). On the other hand, for each fixed \(\ell'\) there is a polynomial-time approximation scheme (PTAS) for \(\ell_p\)-Shortest Path in \(\ell'\) dimensions that runs in time \(n^{O(\ell')}\). We now explain how to combine these two results. Fix $\delta \in (1/\ell, 1/2)$. Divide the coordinates of the cost vectors \(c_e \in \R^\ell\) into \(\ell'= \lceil 1/\delta\rceil\) groups, each of size at most \(k = \lceil \delta\ell \rceil\)
and then add up the coordinates in each group. For each cost vector $c_e \in \mathbb R^\ell$, we obtain a new vector $c_e' \in \mathbb R^{\ell'}$. Costs $c'_e$ approximate costs $c_e$ within a factor of $k^{1-1/p}$ in the following sense: for every path $P$, 
\begin{equation}
    \Bigl\|\sum_{e\in P} c_e\Bigr\|_p\leq \Bigl\|\sum_{e\in P} c_e'\Bigr\|_p \leq k^{1-1/p} \Bigl\|\sum_{e\in P} c_e\Bigr\|_p .\label{eq:remark}
\end{equation}
Using the PTAS, we solve the problem with costs $c'_e$ and by~(\ref{eq:remark}) get a $(1+\varepsilon) k^{1-1/p}$ approximation to the original problem. 
We conclude that there exists an approximation algorithm that finds an $O((\delta\ell)^{1-1/p})$ approximation in time $n^{O(1/\delta)}$ (for every $\delta\in (1/\ell, 1/2)$).
\end{remark}

In the course of analyzing our algorithms, we prove a new majorization inequality for pseudo-expectations (see \Cref{sec:majorization-inequalities}) generalizing pseudo-expectation Lyapunov's  and H\"older inequality (for the latter see~\cite[arXiv version]{BKS14}). We believe this result to be of independent interest.

\paragraph{Hardness results} We also complement the analysis above with several hardness results for $\ell_p$-Shortest Path. First, in \Cref{sec:hardness-from-congestion-minimization}, we give a reduction showing that the problem of congestion minimization can be reduced to the $\ell_\infty$-Shortest Path problem. This simultaneously speaks to the broad applicability of $\ell_p$-Shortest Path and implies hardness for the $\ell_\infty$-version of the problem, following a result of Chuzhoy and Khanna~\cite{chuzhoy2006hardness}. 
\begin{theorem}\label{thm:intro-hardness-from-congestion-minimization}
    The $\ell_\infty$-Shortest Path problem is hard to approximate within an $\Omega(\log n/\log\log n)$-factor unless $\mathrm{NP} \subseteq \mathrm{ZPTIME}(n^{\log \log n})$.
\end{theorem}
This theorem slightly strengthens the $\Omega(\log^{1-\eps} \ell)$-hardness of approximation result by
 Kasperski and Zieli\'nski~\cite{kasperski2009approximability}.
Finally, in \Cref{sec:hardness-from-CVP}, we show that allowing the entries of the cost vectors to be negative makes the problem substantially harder. We do this by giving a reduction from the Closest Vector problem in lattices to this (potentially negative costs) version of the $\ell_p$-Shortest Path problem.
Below, \(\Omega_p\) hides a constant depending on \(p\).
\begin{theorem}\label{thm:intro-hardness-from-CVP}
     For every $p \in [1,\infty]$, it is \(\mathrm{NP}\)-Hard to approximate the $\ell_p$-Shortest Path problem allowing negative edge costs within a factor of $n^{\Omega_p(1 / \log \log n)}$.
\end{theorem}
\begin{remark}
The requirement that all the coordinates of cost vectors $c_e$ are non-negative can be slightly relaxed when $p=2$. For our algorithms to work, it is sufficient that all pairwise inner products of cost vectors are non-negative. That is, instead of requiring that the Gram matrix of cost vectors is completely positive, we can only require that it is doubly non-negative. 
\end{remark}

\paragraph{From shortest path to network design} In network design problems, one is
given a graph \(G = (V, E)\) with non-negative edge costs \(c_e \geq 0\), and wishes to find a subgraph \(F =(V_F, E_F)\) that minimizes the total cost \(\sum_{e \in E_F} c_e \) subject to some connectivity constraints.
By varying the set of allowed subgraphs \(F\), this paradigm encapsulates many central and well-studied network design problems, including the Survivable Network Design, Steiner Forest, Steiner Tree, and Minimum Spanning Tree.
In this paper we explore two network design problems, Group Steiner Tree and Asymetric Traveling Salesperson (ATSP). We first recall the Group Steiner Tree problem.

\begin{problem}[Group Steiner Tree] Given a weighted undirected graph $G = (V,E,c)$, as well as $k$ subsets $R_1, \dots, R_k$ of $V$, find a minimum-cost subtree $T$ of $G$ containing at least one vertex from each $R_i$.
\end{problem}

We then introduce an analogous group variant of ATSP:

\begin{problem}[Group ATSP] Given a weighted directed graph $G=(V,E,c)$ and a collection of subsets $R_1, \dots, R_k$ of \(V\), find a minimum-cost tour that visits at least one vertex in each $R_i$. 
\end{problem}

As in the case of Shortest Path, it is natural to ask whether we can approximately solve $\ell_p$ versions of other network design problems efficiently. Prior work has been done in this area. Hamacher and Ruhe~\cite{hamacher1994spanning} studied $\ell_\infty$-Minimum Spanning Tree, and proved that it is NP-complete. Following that, the complexity of the problem has been nearly completely settled: Chekuri, Vondrák, and Zenklusen~\cite{chekuri2010dependent} presented an $O(\log \ell/ \log\log \ell)$-approximation algorithm, while Kasperski and Zielinski~\cite{kasperski2011approximability} (also see \cite{kasperski2016robust}, Table 1) proved an $\Omega(\log^{1-\varepsilon} \ell)$-hardness of approximation for every $\varepsilon >0$, unless all problems in NP can be solved in quasi-polynomial time. Laddha, Singh and Vempala~\cite{laddha2022socially} studied the $\ell_\infty$-version of a subclass of network design problems which encompasses the Generalized Steiner Network problem, and gave a polynomial-time $\ell$-approximation algorithm for it.




\paragraph{Our results for $\ell_p$-Group Steiner Tree and $\ell_p$-Group ATSP} We consider the $\ell_p$-version of the Group Steiner Tree and the Group ATSP problems.

In \Cref{sec:ATSP-and-Steiner-Tree}, we refine the SoS relaxation for $\ell_p$-Shortest Path to obtain approximation algorithms for $\ell_p$-Group ATSP and $\ell_p$-Group Steiner Tree, and thus obtain approximation algorithms for these problems as well. In particular, we prove the following two results:

\begin{theorem}\label{thm:intro-atsp-guarantees}  
There exists an approximation algorithm for $\ell_p$-Group ATSP that given graph $G$, groups $R_i$, and parameters $p\in\bbZ_{\geq 0}$ and $c\in(0,1/2)$ finds a $c^2 p  \log^{2-1/p} n \log k$ approximation in time $m^{O(p)+ce^{O(1/c)} \log n}= m^{O_c(\log n)}$. We assume that $k$ is at most polynomial in $n$. 
\end{theorem}

\begin{theorem}\label{thm:intro-steiner-tree-guarantees}
There exists an approximation algorithm for
    the $\ell_p$-Group Steiner Tree problem in undirected graphs that given a graph $G$, groups $R_i$, and parameters $p\in\bbZ_{\geq 1}$ and $c\in(0,1/2)$ finds a $c^2 p  \log^{2-1/p} n \log k$ approximation in time $m^{c e^{O(1/c)} \log n}=m^{O_c(\log n)}$. We assume that  $k$ is at most polynomial in $n$. 
\end{theorem}
Note that for $p=\roundup{\log \ell}$, we get an approximation algorithm for the $\ell_\infty$ norm.

\subsection{Related Works}

\paragraph{Previous results on the scalar-cost group Steiner tree problem.}
Group Steiner Tree with scalar costs was introduced by Reich and Widmayer~\cite{reich1989beyond}. Garg, Konjevod, and Ravi~\cite{garg2000polylogarithmic} gave an $O(\log^2 n  \log k)$-approximation to the problem, the first polylogarithmic approximation. Charikar, Chekuri, Goel, and Guha~\cite{charikar1998rounding} gave the same $O(\log^2 n \log k)$-approximation with a deterministic algorithm.
Then Charikar, Chekuri, Cheung, Dai, Goel, Guha, and Li~\cite{charikar1999approximation} gave an \(O(\log^3 k )\) approximation that works even with directed graphs; however, it required quasipolynomial time. Finally, Chekuri and P\'al~\cite{chekuri2005recursive} gave an \(O(\log^2 k )\) approximation for undirected graphs also in quasipolynomial time. The approximation guarantees of \cite{garg2000polylogarithmic} and~\cite{charikar1998rounding} are presented above with the improvement resulting from using metric embedding by Fakcharoenphol, Rao and Talwar~\cite{fakcharoenphol2003tight}. 
In terms of the approximation guarantee and running time, out algorithm is most similar to that by Chekuri and P\'al~\cite{chekuri2005recursive}: when $k = \Theta(n)$, the approximation guarantees match. In terms of techniques used, our algorithm uses some ideas from that by Garg, Konjevod, and Ravi~\cite{garg2000polylogarithmic}.

\paragraph{Dijkstra-style Algorithm for $\ell_p$-Shortest Path} The authors of \cite{bilo2017simple} describe a Dijkstra-style algorithm for the $\ell_p$-Shortest Path problem and claim that it achieves an $O(\min\{p,\log \ell\})$-approximation. However, we show that this claim is incorrect and, in fact, the approximation factor of their algorithm is at least $\Omega(n^{1-1/p})$. We discuss this algorithm in \Cref{sec:dijkstra}.

\paragraph{Multi-objective combinatorial optimization for shortest path and network design} The work in this paper is closely related to multi-objective combinatorial optimization (MOCO). This area studies combinatorial optimization problems in the presence of multiple competing objective functions. Much of the literature on MOCO is concerned with finding all or some Pareto efficient solutions, that is, solutions that are not dominated in every objective by any other solution, a problem which is often intractable due to the exponential number of these points. In particular, there is prior MOCO work on both shortest path~\cite{hansen1980bicriterion,martins1984multicriteria,breugem2017analysis} and network design problems \cite{ravi1993many}. We refer the reader to the paper of Ruzika and Hamacher~\cite{ruzika2009survey} for a survey on multi-objective spanning tree problems, and the book of Ehrgott~\cite{ehrgott2005multicriteria} for an overview of multi-criteria optimization area as a whole.

\subsection{Technical Overview}
Let us first discuss the $\ell_p$-Shortest Path problem. The most basic variant of this problem is when $G$ is a series-parallel graph (see Section~\ref{sec:prelim} for definitions) and $p=2$. The first idea is to write an LP flow relaxation with a convex objective: minimize $\|\sum_{e\in E} c_e x_e\|^2$ subject to the constraint that $(x_e)_{e\in E}$ define a unit flow from $s$ to $t$. However, this LP has an integrality gap of $\sqrt{\ell}$. To deal with this problem, Li, Xu, and Zhang~\cite{li2023polylogarithmic} introduced a new constraint for $\ell_{\infty}$-Shortest Path: the constraint loosely speaking says that the LP cost of every subgraph/block $B$ in the series-parallel decomposition of $G$ is at most $\opt$ times the probability (according to the LP) that the path visits $B$. Unfortunately, this new constraint does not help when $p=2$. 

Instead, we consider a sum-of-squares (SoS) strengthening of this LP.\footnote{It is sufficient to use a vector-flow SDP with one vector variable per edge in order to  approximate the $\ell_2$-cost in series-parallel graphs. However, we need higher degree SoS relaxations when $p >2$ and in general graphs.}   The SoS relaxation gives valuations not only to individual edges but also to tuples of edges. Using the standard notation of pseudo-expectations (see Section~\ref{sec:prelim}), the SoS relaxation for $p=2$ gives $\psE[x_e]$ for every edge $e$ and $\psE[x_{e_1}x_{e_2}]$ for every pair of edges $e_1$ and $e_2$. The former could be interpreted as the probability that $e\in P$ and the latter as the probability as both $e_1,e_2\in P$ according to the relaxation.
To the best of our knowledge, this is the first flow-based SoS or SDP relaxation studied in the literature.

Our algorithm is very straightforward. We start at $u_0=s$, then choose one of the edges outgoing from $u_1$ with probability of choosing $e$ being equal to $\psE[x_e]$. We get to a vertex $u_1$ and then again sample one of the edges leaving $u_1$ with probability of choosing $e$ proportional to $\psE[x_e]$. We repeat this step over and over until we reach $t$. It is clear that the algorithm finds an $s$-$t$ path $P$. 

Now we need to upper bound the cost of $P$. We do that recursively using the series-parallel decomposition of $G$.\footnote{Interestingly, neither the relaxation nor the algorithm uses the series-parallel decomposition of $G$.} Assume that $G$ is composed of subgraphs/blocks $B_1,\dots,B_t$ 
and our algorithm achieves an $\alpha$ approximation for the squared $\ell_2$-cost in each of them. For simplicity, assume that $t=2$ for now. There are two cases: $G$ is a (i) parallel and (ii) series composition of $B_1$ and $B_2$. Consider the first case.
The SoS relaxation ensures that $\psE[x_{e_1}x_{e_2}] = 0$ for all $e_1\in B_1$ and $e_2\in B_2$; this means that the SoS solution is simply a convex combination of solutions for $B_1$ and $B_2$ with some weights $p_1$ and $p_2$.
Also, with probability $p_1$, the first edge of $P$ will be in $B_1$ and then the entire path will be in $B_1$; similarly, with probability $p_2$, the entire path will be in $B_2$. Thus running the algorithm reduces to randomly choosing a block $B_i$ with probability $p_i$ and then running the algorithm in $B_i$. Since the algorithm gets an $\alpha$ approximation in each $B_i$, it also gets an $\alpha$ approximation in the entire graph. Interestingly, this step would already fail if we used the basic LP relaxation; however, a Sherali--Adams or configuration LP would work in \textit{this} case.

The second case -- when $G$ is a series composition of $B_1$ and $B_2$ is more challenging and requires the power of an SDP relaxation. Let $P_i = P\cap B_i$. Write the squared objective as follows:
\begin{equation}\label{eq:overview:cost}\Bigl\|\sum_{e\in P} c_e\Bigr\|_2^2 = \underbrace{\Bigl\|\sum_{e\in P_1} c_e\Bigr\|_2^2 + \Bigl\|\sum_{e\in P_2} c_e\Bigr\|_2^2}_{\leq \alpha \opt^2 \text{ (in expectation)}} + 2\sum_{\substack{e_1\in P_1\\e_2\in P_2}} \langle c_{e_1},c_{e_2}\rangle.
\end{equation}
The first two terms are squared $\ell_2$-costs of paths $P_1$ and $P_2$. As we assumed, they are at most $\alpha$ times their SoS costs, and thus their sum is at most $\alpha \opt^2$ (in expectation). We now analyze the third term. It is not hard to see that our algorithm samples edges in $P_1$ and $P_2$ independently (because the last vertex of $P_1$ and the first vertex of $P_2$ are fixed). Therefore, 
\begin{equation}\label{eq:overview:term-3}
\Exp\Bigl[\sum_{\substack{e_1\in P_1\\e_2\in P_2}} \langle c_{e_1},c_{e_2}\rangle\Bigr] = \sum_{\substack{e_1\in B_1\\e_2\in B_2}} \langle c_{e_1},c_{e_2}\rangle \cdot \Prob{e_1,e_2\in P}=
\sum_{\substack{e_1\in B_1\\e_2\in B_2}} \langle c_{e_1},c_{e_2}\rangle \cdot \Prob{e_1\in P} \cdot \Prob{e_2\in P}.
\end{equation}
It would be natural to upper bound this expression by the corresponding expression in the SoS objective (appropriately scaled):
\[
 \sum_{\substack{e_1\in B_1\\e_2\in B_2}} \langle c_{e_1},c_{e_2}\rangle \cdot \psE[x_{e_1}x_{e_2}].\]
However, this is not possible, since it may happen that  
$\Prob{e_1\in P}\cdot \Prob{e_2\in P} > 0 $ but $\psE[x_{e_1}x_{e_2}] = 0$. Instead, observing that for every edge $e$, $\Prob{e\in P} = \psE[x_e]$, we rewrite and upper bound \eqref{eq:overview:term-3}:
\begin{align*}
\sum_{\substack{e_1\in B_1\\e_2\in B_2}} \langle c_{e_1},c_{e_2}\rangle  \psE[x_{e_1}] \psE[x_{e_2}]
&\leq 
\sum_{e_1,e_2\in E} \langle c_{e_1},c_{e_2}\rangle \cdot \psE[x_{e_1}] \cdot \psE[x_{e_2}]
\\
&=\Bigl\|\psE\Bigl[\sum_{e\in E}c_{e}x_{e}\Bigr]\Bigr\|_2^2
\leq \psE\Bigl[\Bigl\|\sum_{e\in E} c_e x_e\Bigr\|^2\Bigr]\leq \opt^2.
\end{align*}
Here, we first expanded the summation, then used the pseudo-expectation Lyapunov's inequality $\|\psE[f]\|_2^2 \leq \psE[\|f\|_2^2]$ (see Fact~\ref{claim:Lyapunov}), and finally observed that the last pseudo-expectation is the SoS objective for $G$.
We conclude that the expected squared cost of $P$ is at most $(\alpha+2) \opt$.
Applying this argument recursively, we get an $O(d)$-approximation for the squared cost and an $O(\sqrt{d})$-approximation for the cost itself in series-parallel graphs of order/depth $d$.

When $p > 2$ and blocks in the series-parallel composition of $G$ are formed by $t > 2$ lower-order blocks, the proof becomes more technical. In particular, we need to use a new majorization inequality for pseudo-expectation, which we present in Section~\ref{sec:majorization-inequalities}. 

The SoS relaxation for arbitrary graphs is the same as that for series-parallel graphs (except that its degree is higher). However, the rounding algorithm is quite different. Very informally, the algorithm in its simplest form resembles Savitch's algorithm for $s$-$t$ connectivity in $O(\log^2 n)$ space~\cite{savitch1970} (see also~\cite{chekuri2005recursive}): (i) we sample the middle edge $e = (u,v)$ of the path using probabilities provided by $\psE[\cdot]$, (ii) condition $\psE[\cdot]$ on $e$ being the middle edge, (iii) then recursively find paths $P_1$ from $s$ to $u$ and (independently) $P_2$ from $v$ to $t$, using the \textit{conditional} pseudo-expectation. To upper bound the cost, as in the analysis of the algorithm for series-parallel graphs, we first use \eqref{eq:overview:cost} , then bound the third term using a variant of \eqref{eq:overview:term-3},  and finally use the majorization inequality for pseudo-expectations.

To solve Group ATSP, we loosely speaking add SoS constraints that require that the tour $P$ visits every group (for technical reasons, we need to require that $P$ visits each group exactly once).  Then we run the rounding algorithm for $\ell_p$-Shortest Path in arbitrary graphs. It is not guaranteed that $P$ indeed visits every group; however, using the machinery we developed for bounding the cost of $P$, we show that $P$ visits every group with probability at least $\Omega(1/\log n)$. By sampling sufficiently many tours and concatenating them, we obtain the desired solution with high probability.  The Group Steiner Tree problem easily reduces to Group ATSP.

\subsection{Paper Organization}
The rest of the paper is organized as follows.
In \Cref{sec:prelim} we define series-parallel graphs, relevant combinatorial quantities and notation used in the rest of the paper, and give some basic facts on Sum-of-Squares relaxations.
In \Cref{sec:sos-relaxation}, we describe our Sum-of-Squares relaxation for \(\ell_p\)-Shortest Path in directed acyclic graphs.
In \Cref{sec:majorization-inequalities} we show a majorization inequality for pseudo-expectations used in the analysis of our algorithms.
In \Cref{sec:rounding} we describe and analyze our rounding algorithm for \(\ell_p\)-Shortest Path in series-parallel graphs.
In \Cref{sec:general} we describe our approximation algorithm for \(\ell_p\)-Shortest Path in arbitrary graphs.
In \Cref{sec:ATSP-and-Steiner-Tree}, we present our algorithms for \(\ell_p\)-Group ATSP and \(\ell_p\)-Group Steiner Tree.
In \Cref{sec:all-hardness-results}, we give our hardness results: hardness of approximation results for \(\ell_p\)-Shortest Path with potentially negative edge costs and for \(\ell_\infty\)-Shortest Path. We also show that our analysis of the \(\ell_p\)-Shortest Path algorithm in series-parallel graphs is tight.
In \Cref{sec:missing-proofs}, we prove a recurrence formula and upper bound on multidimensional Bell numbers, which are used in the analyses of our algorithms for \(\ell_p\)-Shortest Path.

\section{Preliminaries and Notation} \label{sec:prelim}
In this paper, all \(\log\)s are base \(2\).
In this paper, we consider Shortest Path and Group ATSP in directed graphs and Group Steiner Tree in undirected graphs. We assume that graphs may have parallel edges. 
Let $G=(V, E)$ be a directed graph. We denote $n=|V|$ and $m=|E|$.
For \(v \in V\), denote the sets of its outgoing and incoming edges by \(\delta^+(v)\) and \(\delta^-(v)\), respectively. Similarly, define \(\delta^+(A)\) and \(\delta^-(A)\) for subsets of vertices $A$. Finally,  denote the set of edges from $A$ to $B$ by $\delta(A,B)$. We denote the $i$-th coordinate of vector edge cost $c_e$ by $c_e(i)$.
\subsection{Series-Parallel Graphs}

We start with providing a recursive definition of directed series-parallel graphs with source $s$ and sink $t$. 
A graph on two vertices $s$, $t$ and one or more edges from $s$ to $t$ is a series-parallel graph of order (depth) $0$. We denote the order of $G$ as $\sheight(G)$.
\begin{itemize}
\item \textbf{Parallel Composition.} Let $B_1$,\dots, $B_t$ be series-parallel graphs that share only vertices $s$ and $t$. Then their union $G$ is a series-parallel graph. Define $\sheight(G) = \max_j \sheight(B_j)$.
\item \textbf{Series  Composition.}
Let $B_1$,\dots, $B_t$ be series-parallel graphs. Denote the source and sink of $B_i$ by $s_i$ and $t_i$ (respectively). Assume that $t_i = s_{i+1}$ for all $i\in\{1,\dots, t-1\}$ and that graphs $B_i$ do not share any other vertices. Then the union $G$ of graphs $B_i$ is a series-parallel graph. Define $\sheight(G) = \max_j \sheight(B_j) + 1$.
\end{itemize}
In this definition, we only count \textit{series}  compositions when we compute the order of a series-parallel graph. We  call vertices $s$ and $t$ \emph{terminals}. We call intermediate graphs that we obtain while constructing $G$ \emph{blocks}. We denote the source and sink of a block $B$ by $s_B$ and $t_B$, respectively. 


\subsection{Combinatorics}
\label{sec:prelim:comb}
\paragraph{Unlabeled Partitions} We say that a tuple of integers $\lambda= (\lambda_1,\dots,\lambda_k)$ is an unlabeled partition of an integer $n \geq 1$ if $n = \sum_{i=1}^k \lambda_i$ and $\lambda_1\geq \lambda_2 \geq \cdots \ge \lambda_k \geq 1$. We will denote this by $\lambda \prt n$. We will denote the length of $\lambda = (\lambda_1,\dots, \lambda_k)$ by $|\lambda| = k$.

Given an \(n\) and some tuple of non-negative integers \(\alpha\) with \(\alpha_1 + \ldots + \alpha_k = n\), we use standard notation
for the multinomial coefficient
\[\binom{n}{\alpha} \defeq \frac{n!}{\prod_{i=1}^k \alpha_i!}\]




\paragraph{Multidimensional Bell Numbers}
Recall that the \(n\)th Bell number $\bell_n$ equals the number of labeled partitions of a set of size $n$. In this paper, we will need a generalization of Bell numbers, known as multidimensional Bell numbers (see~\cite[Example 5.2.4]{stanley2023enumerative} and \cite{de2011set}).
\begin{definition}
We say that a collection of subsets $P$ is a partition of a set $S$ if all subsets in $P$ are disjoint and their union is $S$.
Consider two partitions $P$ and $P'$ of $S$. We say that $P'$ is a refinement of $P$ if every $A\in P'$ is a subset of some $B\in P$. 

A $d$-dimensional partition of $p$ is a tuple $(P_1,\dots, P_d)$ where all $P_i$ are partitions of $[p] \eqdef \{1,\dots, p\}$ and each $P_{i+1}$ is a refinement of $P_i$. 
The $d$-dimensional Bell number $\bell_{d}(p)$ is the number of $d$-dimensional partitions of $p$. If $d=0$ or $p=0$, we let  $\bell_d(p) \eqdef 1$.
\end{definition}
Note that 1-dimensional Bell numbers are simply the standard Bell numbers: $\bell_1(i) = \bell(i)$. We can also restate the definition of $\bell_d(n)$ as follows. $\bell_d(n)$ is the number of $(d+2)$-level rooted trees with $n$ labeled leaves: the root must be in level \(0\), all leaves must be in level $d+1$, and all leaves are labeled with numbers from $1$ to $n$ with each number being used exactly once.

We will need the following recurrence formula for $d$-dimensional Bell numbers, which is proved in \Cref{sec:missing-proofs}.

\begin{restatable}{lemma}{bellRecurrence}
\label{claim:bell-recurrence}
For every $d\geq 1$ and $p \geq 1$, we have
$$
\bell_d(p) = \sum_{\lambda \prt p} \binom{p}{\lambda} \prod_{i=1}^{|\lambda|} \bell_{d-1}(\lambda_i)\left/ \prod_{j=1}^p \cnt(j, \lambda)!\right.
$$
where $\cnt(j, \lambda)$ is the number of times $j$ appears in $\lambda$.

\end{restatable}

Now, we describe the exponential generating function for sequence $(\bell_d(i))_i$ when $d$ is fixed. 
\begin{fact}\label{fact:gen}\cite[Example 5.2.4]{stanley2023enumerative}
Let $f_0(x) = \exp(x)$ and $f_{i+1}(x) = \exp(f_i(x) -1)$. Then the exponential generating function for sequence $(\bell_d(i))_{i=0}^\infty$ is given by:
\begin{equation}\label{eq:bell:recurrence}
\sum_{i=0}^\infty \frac{\bell_{d}(i) x^i}{i!} = f_d(x).
\end{equation}
\end{fact}
In this paper, we will present an approximation algorithm for $\ell_p$-Shortest Path in depth-$d$ series-parallel graphs with approximation factor $\abell_d(p) \defeq \bell_d(p)^{1/p}$. From Fact~\ref{fact:gen}, we obtain the following upper bound on $\abell_d(p)$, proved in Appendix~\ref{sec:missing-proofs}.

\begin{restatable}{claim}{bellNumberAsymptotics}\label{claim:bellNumberAsymptotics}
 For all $p \geq 1$ and $d\geq 1$, we have
$$\abell_d(p) = \bell_d(p)^{1/p} = O(pd^{1-1/p}).$$
Let $\log^{(j)} p = \underbrace{\log \cdots \log}_{j\text{ times}} p$ and $\log^* p $ be the largest value of $j$ such that $\log^{(j)} p \geq 1$. Then, the following upper bound on $\abell_d(p)$ holds for $d\leq \log^{*} p:$
$$\abell_d(p) = \bell_d(p)^{1/p} \leq O\Bigl(\frac{p}{\log^{(d)} p}\Bigr).
$$
\end{restatable}

\paragraph{Poisson branching process}
Let $\Pois(t)$ be the Poisson distribution with rate $t$. We define the Poisson branching process $Z_0, Z_1,\dots, Z_d, \dots$ as follows. Let $Z_0 = 1$ (always). Let $Z_1$ be sampled from $\Pois(Z_0) = \Pois(1)$. Then, let $Z_2$ be sampled from $\Pois(Z_1)$ and so on; let $Z_{i+1}$ be sampled from $\Pois(Z_{i})$. We refer to the distribution of $Z_d$ (for a fixed $d$) as an iterated Poisson distribution. We will use the following fact.
\begin{fact}\label{fact:iterated-poisson-moments} (See~\cite{de2011set}) Let $d \geq 0$ and $p\geq 1$, then
$$\E{Z_d^p} = \bell_d(p).$$
\end{fact}

\subsection{Sum-of-Squares Relaxations}
We recall some basics about the sum-of-squares relaxations. Sum-of-squares relaxations can be thought of in terms of moment matrices, pseudo-distributions, and pseudo-expectations. As is common, we will use the pseudo-expectation framework in this paper. We refer the reader to \cite{FKP19} for a detailed description of the sum-of-squares framework.

Consider a set of variables $x_i$ where $i$ belongs to some set of indices $\calI$.  We denote the entire collection of all variables $(x_i)_{i\in \calI}$ by $\x$. Consider the set of multivariate polynomials $\bbR_{\leq d} [\x] \eqdef \bbR_{\leq d} [\{x_i:i\in \calI\}]$ in variables $x_i$ of degree at most $d$. We say that $f\in\bbR_{\leq d} [\x]$ is a sum of squares (SoS) if \(f = \sum_{i=1}^m f_i^2\) for some polynomials $f_1,\dots, f_m$. Note that the product of SoS polynomials is a SoS, and so is any linear combination of SoS polynomials with positive coefficients.

In this paper, we consider SoS relaxations for the Boolean hypercube; that is, all variables $x_i$ take values $0$ and $1$ in the intended solution. Therefore, we work with the quotient ring $\R_{\leq d}/\langle x_i^2 -x_i\rangle_i$, where $\langle x_i^2 -x_i\rangle_i$ is the ideal generated by polynomials $x_i^2 - x_i$. In other words, we identify monomials $x_{i_1}^{a_1}, \dots, x_{i_t}^{a_t}$ and $x_{i_1},\dots, x_{i_t}$ for all $i_1,\dots, i_t$ and $a_1,\dots, a_t \geq 1$ such that $\sum_{i=1}^t a_i \leq d$. In particular, we will write $f=g$ if $f-g\in \langle x_i^2 -x_i\rangle_i$.

\begin{definition}
A linear map $\psE: \bbR_{\leq d} [\x]/\langle x_i^2 -x_i\rangle_i \to \bbR$ is a pseudo-expectation of degree $d$ if it satisfies the following properties.
\begin{itemize}
\item $\psE[1] = 1$,
\item $\psE[f^2] \geq 0$ for every polynomial $f$ of degree at most $d/2$,
\end{itemize}
We say that a pseudo-expectation $\psE$ satisfies an equality constraint $f = 0$ if $\psE[fg] = 0$ whenever $\deg fg \leq d$.
\end{definition}

Given an objective function $f$ and sets of equality and inequality constraints, we can find a pseudo-expectation $\psE$ that maximizes $\psE[f]$ and satisfies all the constraints in time polynomial in $N^{O(d)}$, where $N$ is the number variables, as long as it satisfies certain regularity conditions~\cite{RW17}. When we prove any statements about pseudo-expectations $\psE[f]$ below, we will \textit{always} implicitly assume that $d \geq \Omega(\deg f)$ so that all the inequalities appearing in the proofs have degree at most $d$.

\begin{definition}
Let $\psE$ be a pseudo-expectation of degree $d$.
Assume that $g$ is a sum of squares and $\psE[g] > 0$. Then the conditional pseudo-expectation $\psE[\cdot \given g]$ operator is defined as follows: $\psE[f \given g] \defeq \psE[fg]/\psE[g]$.
\end{definition}

\begin{fact}\label{fact:cond-pseudo-is-pseudo}
A conditional pseudo-expectation $\psE[\cdot \given g]$ is a pseudo-expectation of degree $d' = d - \deg g$. If $\psE$ satisfies an equality or inequality constraint of degree at most $d'$, then so does $\psE$.
\end{fact}

We will use Lyapunov's inequality for pseudo-expectations (which is also referred to as Jensen's inequality in the literature).
\begin{claim}
\label{claim:Lyapunov}
Let \(g\) be a sum of squares and f be any polynomial. Assume that $\deg f^2g \leq d$.  Then,
\begin{equation}\label{eq:Lyapunov}
\psE[fg]^2 \le \psE[f^2g]\psE[g].
\end{equation}
If $\psE[g] > 0$, the inequality can be restated as
\begin{equation}\label{eq:condvar_ineq}
\psE[f\given g]^2 \le \psE[f^2\given g].
\end{equation}
\end{claim}
\begin{claim}
\label{claim:sum-f-ineq}
Let $f_1,\dots, f_t$ be SoS polynomials. Then, $\psE[\left(\sum_{i=1}^t f_i\right)^p] \geq \sum_{i=1}^t \psE[f_i^p]$.
\end{claim}
\begin{proof}
We expand $\left(\sum_{i=1}^t f_i\right)^p$ as $\sum_{\substack{\alpha_1, \ldots, \alpha_t \ge 0 \\ \alpha_1 + \cdots + \alpha_t = p}} \binom{p}{\alpha}f_1^{\alpha_1}\cdots f_t^{\alpha_t}$. All terms in the expansion are SoS polynomials and thus have non-negative pseudo-expectations. The claim follows from the observation that all terms $f_i^p$ are present in the expansion.
\end{proof}

\section{Sum of Squares Relaxation for \texorpdfstring{$\ell_p$}{l-p}-Shortest Path}\label{sec:sos-relaxation}
In this section, we first present our SoS relaxation for $\ell_p$-Shortest Path in directed acyclic graphs (DAGs). In \Cref{sec:rounding}, we will present a rounding algorithm for series-parallel graphs and then, in \Cref{sec:general}, for  layered graphs. The latter result will also yield an algorithm for arbitrary graphs. We will also describe a few basic properties that feasible solutions for this relaxation satisfy.
\paragraph{Relaxation}
We use a degree $2p$ SoS relaxation with variables $\x = \p{x_e}_{e\in E}$ for $\ell_p$-Shortest Path in series-parallel graphs. 
\begin{align*}
    \text{min }\quad& \psE\Bigl[\sum_{i=1}^\ell \Bigl(\sum_e c_e(i) x_e\Bigr)^p\Bigr]\\
    \text{subject to }\quad&
    (x_e)_{e\in E} \text{ is a unit flow from } s \text{ to } t
\end{align*}
The flow constraint says that $\sum_{e\in \delta^+(u)} x_e - \sum_{e\in \delta^-(u)} x_e = 0$ for all $u$ other than $s$ and $t$ (flow conservation) and $\sum_{e\in \delta^+(s)} x_e - 1 = 0$ ($x_e$ sends 1 unit of flow from $s$ to $t$).
It is clear that this is a relaxation for the $\ell_p$-Shortest Path problem: $\psE[\sum_{i=1}^\ell (\sum_e c_e(i) x_e)^p\Bigr] \leq \opt^p$, where $\opt$ is the $\ell_p$-cost of the optimal $s$-$t$ path.

\paragraph{Basic properties of the SoS relaxation}
We say that two edges $e_1$ and $e_2$ are \emph{compatible} if both of them belong to some $s$-$t$ path; otherwise, we say that $e_1$ and $e_2$ are \emph{incompatible}. In a series-parallel graph edges $e_1$ and $e_2$ are incompatible if and only if there exist two parallel blocks \(B_1\) and \(B_2\) such that \(e_1\) lies in $B_1$ and \(e_2\) lies in \(B_2\).
For any set of vertices \(A\), let $\x^{+}_A = \sum_{e\in \delta^+(A)} x_{e}$ and  $\x^{-}_A = \sum_{e\in \delta^-(A)} x_{e}$.
\begin{claim} \label{claim:flow-basics}
Assume that $G$ is a DAG and $\psE$ is a feasible pseudo-expectation for the relaxation. Let $h$ be a multivariate polynomial. Then
\begin{enumerate}
\item If $A \subseteq V$ contains neither of the terminals, then $\psE[(\x^{+}_A - \x^{-}_A)h] = 0$. If $A$ contains $s$ but not $t$, $\psE[(\x^{+}_A - \x^{-}_A)h] = \psE[h]$.
\item If $e_1$ and $e_2$ are not compatible, then $\psE[x_{e_1}x_{e_2}h] = 0$.
\item Assume further that $G$ is a series-parallel graph. Let $(L, R)$ be an $s_B$-$t_B$ cut in a block $B$. Let $f_{LR} = \sum_{e\in \delta(L,R)} x_e$. 
 Then 
$$\psE\Bigl[f_{LR}h\Bigr] = \psE\Bigl[\bigl(\sum_{e\in \delta^+(s_B) \cap B} x_e\bigr)h\Bigr].$$
In particular, $\psE[f_{LR}h]$ does not depend on the cut $(L,R)$ in $B$.
\end{enumerate}
\end{claim}
\begin{proof}
1. The SoS relaxation satisfies the flow conservation constraints and the constraint that the amount of flow being routed equals $1$. Therefore, it satisfies any linear combination of them. In particular, it satisfies degree-1 polynomial equations $\x^{+}_A - \x^{-}_A = 0$ when $s,t\notin A$ and $\x^{+}_A - \x^{-}_A = 1$ when $s\in A$ but $t\notin A$. The first item follows.

\noindent 2. 
It is sufficient to prove the statement for all monomials (the claim then follows by the linearity of $\psE$). Thus, we will assume that $h$ is a monomial.
Recall that an $s$-$t$ cut is monotone if it cuts exactly one edge on every $s$-$t$ path. Since $e_1$ and $e_2$ are incompatible, there is a monotone $s$-$t$ cut $(A,\bar A)$ that cuts both of them. We apply item 1 to polynomial $x_{e_1} h$ and get
$\psE[x_{e_1} h] = \psE[x_{e_1} \x^+_A h] = \psE[x_{e_1} (x_{e_1} + x_{e_2} + \dots) h]$. Here, $\dots$ is a sum of some $x_e$ (the coefficient of each of them is 1). Note that $x_e h$ is a monomial, thus $x_e h = (x_e h)^2$ and therefore $\psE[x_eh] \geq 0$. We conclude that 
$$\psE[x_{e_1} h] \geq \psE[x_{e_1} (x_{e_1} + x_{e_2}) h] = 
\psE[(x_{e_1} + x_{e_1}x_{e_2}) h]$$ and simplifying, we get \(\psE[x_{e_1}x_{e_2}h] \leq 0\).
Since $x_{e_1}x_{e_2}h$ is a monomial, $\psE[x_{e_1}x_{e_2}h]\geq 0$, and thus $\psE[x_{e_1}x_{e_2}h] = 0$.

\noindent 3. Item 3 follows immediately from item 1.
\end{proof}

\section{Majorization Inequalities for Pseudo-expectations}\label{sec:majorization-inequalities}
In this section, we will prove a majorization inequality for pseudo-expectations. This inequality generalizes already known pseudo-expectation Lyapunov's (see Claim~\ref{claim:Lyapunov}) and H\"older's (see~\cite{BKS14}) inequalities. 

\begin{definition}
    Consider two integer sequences \(a_1 \ge a_2 \ge \ldots \ge a_k \ge 0\)
    and \(b_1 \ge \ldots \ge b_k \ge 0\).
    We say \(a\) majorizes \(b\) and write \(a \succeq b\), if \(\sum_{i=1}^k a_i = \sum_{i=1}^k b_i\) and for all \(1 \le i \le k\) we have
    \[a_1 + \ldots + a_i \ge b_1 + \ldots + b_i.\]
\end{definition}
Sequence majorization is a powerful tool for proving inequalities; it appears in widely used Muirhead's~\cite{M1902} and Karamata's~\cite{K1932}  majorization inequalities. We now present a majorization inequality for pseudo-expectations. An analogous inequality for true expectations easily follows both from Karamata's and from Muirhead's majorization inequality. 
\begin{lemma}\label{lemma:ps_major}
Consider a degree $d$ pseudo-expectation $\psE$. Let \(a \succeq b\)  and \(f\) be an SoS polynomial of degree \(\deg(f^{a_1}) \le d\). Then
    \begin{equation}\label{eq:ps_major}
    \prod_{i=1}^k \psE[f^{a_i}] \ge \prod_{i=1}^k \psE[f^{b_i}].
    \end{equation}
\end{lemma}
\begin{proof}

First, observe that this inequality for sequences $(r+1,r-1)\succeq (r,r)$ follows from Lyapunov's inequality for pseudo-expectations (Claim~\ref{claim:Lyapunov}). Indeed, let $g = f^{r-1}$. Then, Lyapunov's inequality states that 
\begin{equation}
\label{eq:major_prim}
\psE[f^r]^2 = \psE[fg]^2 \leq \psE[f^2g]\psE[g] = \psE[f^{r+1}]\psE[f^{r-1}],
\end{equation}
as required.
Now  we use this inequality to show the desired inequality \eqref{eq:ps_major} for a more general case \((p + 1, q-1) \succeq (p, q)\).
\begin{claim}\label{claim:ps_major_sub}
    Let \(p \geq q \ge 1\) be integers.
    Then
    \[ \psE[f^{p+1}] \psE[f^{q-1}] \ge \psE[f^p] \psE[f^q].\]
\end{claim}
\begin{proof} We just proved the inequality when $p=q$. So we will assume below that $p > q$.  Since $f$ is an SoS polynomial, $\psE[f^r] \geq  0$ for all integers $0\leq r\leq p$.
Let us assume first that the inequality is strict for all $r$: $\psE[f^r] > 0$. Then, by dividing and multiplying $\psE[f^q] \psE[f^p]$ by $\prod_{r=q}^{p-1} \psE[f^r]$, we obtain the following identity.
\[
            \psE[f^q] \psE[f^p] = \frac{\psE[f^q] \psE[f^{q}] \psE[f^{q+1}] \cdots \psE[f^{p-1}] \psE[f^p]}{\psE[f^q] \psE[f^{q+1}] \cdots \psE[f^{p-1}]} 
\]
Now we upper bound the numerator of this fraction by iteratively applying \eqref{eq:major_prim}.
\begin{align*}
\Bigl(\psE[f^q] \psE[f^{q}]\Bigr) \psE[f^{q+1}] \cdots \psE[f^{p-1}] \psE[f^p] &\leq
\Bigl(\psE[f^{q-1}] \psE[f^{q+1}]\Bigr) \psE[f^{q+1}] \cdots \psE[f^{p-1}] \psE[f^p] \\
\psE[f^{q-1}] \Bigl(\psE[f^{q+1}] \psE[f^{q+1}]\Bigr) \cdots \psE[f^{p-1}] \psE[f^p]  &\leq 
\psE[f^{q-1}] \Bigl(\psE[f^{q}] \psE[f^{q+2}]\Bigr) \cdots \psE[f^{p-1}] \psE[f^p] \\
            &\vdots \\
\psE[f^{q-1}] \psE[f^{q}] \psE[f^{q+1}] \cdots \Bigl(\psE[f^{p}] \psE[f^{p}]\Bigr) &\le 
\psE[f^{q-1}] \psE[f^{q}] \psE[f^{q+1}] \cdots \Bigl(\psE[f^{p-1}] \psE[f^{p+1}]\Bigr)
\end{align*}
We conclude that 
$$\psE[f^q] \psE[f^p]  \leq \psE[f^{q-1}] \psE[f^{p+1}],$$
as desired.
If $\psE[f^r] = 0$ for some $r$, we apply the inequality to $\hat f = f + \varepsilon$ (where $\varepsilon > 0$). Now $\psE[{\bar f}^r] \geq \varepsilon^r > 0$, since $f$ is a SoS polynomial. Therefore,
$\psE[\hat{f}^q] \psE[\hat{f}^p]  \leq \psE[\hat{f}^{q-1}] \psE[\hat{f}^{p+1}]$.
Letting $\varepsilon \to 0$, we obtain the desired inequality in the limit.
\end{proof}

Consider an integer sequence $a_1\geq \dots a_k \geq 0$ and two indices $1\leq i^* < j^*\leq k$ such that $a_{i^*} - a_{j^*} \geq 2$. Define a \textit{transfer} or \textit{T-transform} as follows: we decrease $a_{i^*}$ by 1, increase $a_{j^*}$ by 1, and then sort the obtained sequence in descending order. \Cref{claim:ps_major_sub} implies that \Cref{lemma:ps_major} holds for sequence $a$ and sequence $b$ obtained from $a$ by a T-transform. 

Finally, we use that if $a\succeq b$ then $b$ can be obtained by a sequence of T-transforms~\cite{M1902}: $a = a^{(0)} \mapsto a^{(1)} \mapsto a^{(2)}\mapsto \dots \mapsto a^{(T)} = b$. As we proved, the value of the product $\prod_{i=1}^k \psE[f^{a^{(t)}_i}]$ may only decrease each time we apply a T-transform.
This concludes the proof of the lemma.
\end{proof}

Importantly, conditional pseudo-expectations are pseudo-expectations (see \Cref{fact:cond-pseudo-is-pseudo}) and thus Lemma~\ref{lemma:ps_major} holds for conditional pseudo-expectations as well.

\section{Sum-of-Squares Relaxation Rounding}\label{sec:rounding}
In this section, we describe and analyze a rounding algorithm for  series-parallel graphs. It gives an $(1+\varepsilon) A_d(p) = O(pd)$ approximation for $\ell_p$-Shortest Path in series-parallel graphs of order $d$. 
Later we use a different algorithm with a similar analysis to solve the problem in layered and arbitrary graphs.

\subsection{Algorithm}
Let us denote $p_e = \psE[x_e]$ and $p_u = \sum_{e\in\delta^+(u)} \psE[x_e]$. Note that the SoS relaxation constraints ensure that $p_e$ is an $s$-$t$ flow; $p_u$ equals the amount of flow that leaves vertex $u$. The total amount of flow is 1.

\begin{algorithm}[H]
\caption{Rounding algorithm for SoS}\label{alg:rounding}
\begin{algorithmic}[1]
\State \textbf{Input:} series-parallel graph $G$ with source $s$ and sink $t$, a pseudo-expectation $\psE$ 
\State \textbf{Output:} an $s$-$t$ path in $G$
\State Let $u=s$ and $P$ be an empty path.
\While{\(u \ne t\)}
    \State Sample \(e\in \delta^+(u)\) with probability \(\frac{p_e}{p_u} = \frac{\psE[x_e]}{\sum_{e\in\delta^+(u)} \psE[x_e]}\)
    \State Append $e$ to path $P$
\EndWhile
\State \Return \(P\)\label{alg_line:rounding_sp_return}
\end{algorithmic}
\end{algorithm}

\begin{lemma}\label{lem:basic_prob}
Let $P$ be the path returned by Algorithm~\ref{alg:rounding}.
Then $\Prob{e\in P} = p_e$ for every edge $e$ and $\Prob{u\in P} = p_u$ for every vertex $u$.
\end{lemma}
\begin{proof}
We consider all vertices in topological order and prove the desired formulas for $\Prob{u \in P}$ and $\Prob{(u,v)\in P}$ by induction. For $u=s$, we have $\Prob{s\in P} = 1 = p_s$. Then, $\Prob{(s,v)\in P} = p_{(s,v)}/p_s = p_{(s,v)}$, as required.
Now assume that we proved the formulas for vertices $u'$ proceeding $u$ in the topological order. We have,
$$\Prob{u\in P} = \sum_{e=(u',u)\in \delta^-(u)} \Prob{e\in P} = \sum_{e=(u,u')\in \delta^-(u)} p_{e} = p_u.
$$
Here, we used the induction hypothesis and the flow conservation condition at vertex $u$. Now let $e=(u,v)$.
$$\Prob{e\in P} = \Prob{e \in P \given u\in P} \Prob{u \in P} = (p_e/p_u) \cdot p_u.
$$
\qedhere
\end{proof}

Now we will prove an upper bound on the $\ell_p$-cost of path $P$. The proof will be by induction on the series-parallel decomposition of $G$, going from lower to higher order blocks $B$. To analyze different blocks $B$, we first introduce some relevant notation.

Let us say that a path $P$ visits block $B$ if it contains at least one edge from $B$. 
Let $P_B = P\cap B$ be the restriction of $P$ to $B$. If $P$ does not visit $B$, let $P_B = \varnothing$.
Note that if $P$ visits $B$ then it must go through the source $s_B$ and sink $t_B$ of $B$. However, if \(B\) has a parallel block $B'$, a path may go through $s_B$ and $t_B$ but visit $B'$ rather than $B$ itself.
It follows from Lemma~\ref{lem:basic_prob} that the probability that $P$ visits $B$ equals $p_B \eqdef \sum_{e\in \delta^+(s_B) \cap B} p_e$.

Now, we define conditional expectations and pseudo-expectations restricted to $B$ (that is, conditioned on the event that \(P\) visits \(B\)). 
Let $h_B = \sum_{e \in \delta^+(s_B)\cap B} x_e = \sum_{e \in \delta^+(s_B)\cap B} x_e^2$ be a SoS indicator of the event that $P$ visits $B$. We let
\[\bE_B[\cdot] \eqdef \E{\cdot \given P \text{ visits } B} \qquad\text{and}\qquad\psE_B[\cdot] \eqdef \psE\pb{\cdot \given h_B}.\]

In the sequel, we shall bound the costs of \(P\) coordinate-by-coordinate. Thus, we consider a set of scalar non-negative edge weights \(a_e \ge 0\). Define 
$f_B \eqdef \sum_{e \in B} a_e \cdot x_e$. For a path $P'$, let $\cost(P') = \sum_{e\in P'}a_e$. Note that $f_B = \sum_{e\in B} a_e \cdot x_e^2$ and thus is a sum of squares.

\begin{claim}\label{claim:conditional}
Let $B'$ be a block inside $B$ (possibly $B'= B$). Assume $p_B > 0$. Then, we have $\EE{B}{\cost(P_{B'})^r} = \frac{\E{\cost(P_{B'})^r}}{p_B}$ and 
$\psE_B[f_{B'}^r] = \frac{\psE[f_{B'}^r]}{p_B}$ when $r\in\{1,\dots,p\}$.
\end{claim}
\begin{proof}
We have
\begin{align*}
\E{\cost(P_{B'})^r} &= \E{\cost(P_{B'})^r\given P_B \neq \varnothing}\cdot \Prob{P_B \neq \varnothing} + \E{\cost(P_{B'})^r\given P_B = \varnothing}\cdot \Prob{P_B = \varnothing} \\
&= \EE{B}{\cost(P_{B'})^r} \cdot p_B + 0\cdot (1-p_B) = p_B \cdot \EE{B}{\cost(P_{B'})^r},
\end{align*}
as required. To prove the second identity, consider a monomial $g$ in the expansion of $f_{B'}^r$. We now prove that $\psE[gh_B] = \psE[g]$ and thus $\psE_B[g] \eqdef \psE[gh_B]/p_B = \psE[g] / p_B$.
Note that only $x_e$ with $e \in B'$ appear in $g$, and $\deg g = r \geq 1$. Choose an arbitrary $x_e$ in $g$, say $e=(u,v)$ and let \(g'\) be such that $g = g'x_e$. Let $(L,R)$ be a monotone $s_B$-$t_B$ cut in $B$ that cuts $e$. By Claim~\ref{claim:flow-basics}, item~3,
$$\psE[gh_B] = \psE\Bigl[g\sum_{e'\in\delta(L,R)} x_{e'}\Bigr] = \psE\Bigl[g'\sum_{e'\in\delta(L,R)} x_{e}x_{e'}\Bigr].$$
Now all edges $e'\in \delta(L,R)$ other than $e$ are incompatible with $e$; for them, $\psE[g'x_{e}x_{e'}] = 0$ by Claim~\ref{claim:flow-basics}, item 2. Also, $\psE[g'x_{e}x_{e'}] = \psE[g'x_{e}]$ for $e'=e$. Therefore, 
$\psE[gh_B] = \psE[g'x_{e}] = \psE[g]$,
as required.
\end{proof}

\begin{lemma}\label{lem:cost-analysis}
Let $\psE$ be a feasible solution for the SoS relaxation of $\ell_p$-Shortest Path. Let $B$ be a block of order $\h=\sheight(B)$ with $p_B > 0$. Then for every $r\leq p$, we have
$$
\EE{B}{\cost(P_B)^{r}} 
\leq \bell_\h(r) \psE_B[f_B^{r}] .
$$
Here, $\bell_{\h}(r)$ is an $\h$-dimensional Bell number (see Section~\ref{sec:prelim:comb} for details).
\end{lemma}
\begin{proof}
We will prove the upper bound on $\E{\cost(P_B)^r}$ by induction on \(r\) and on the series-parallel decomposition of $B$. 
If $r=0$ or $\h=0$, the claim trivially holds.
We consider two cases: when a block $B$ is a parallel composition and when it is a series composition of lower-level blocks. We start with the former, much simpler case when $B$ is a parallel composition of blocks $B_1,\dots, B_t$ sharing the same source $s_B$ and sink $t_B$. If $P$ visits $B$, then it visits exactly one of the blocks $B_1, \dots, B_t$. Therefore,
\begin{equation}\label{eq:parallel:E}
\EE{B}{\cost(P_B)^r} = \sum_{i=1}^t\EE{B_i}{\cost(P_{B_i})^r} \Prob{P \text{ visits } B_i | P \text{ visits } B} = \sum_{i=1}^t\frac{p_{B_i}}{p_B}\EE{B_i}{\cost(P_{B_i})^r}. 
\end{equation}
Note that $f_B = \sum_{i=1}^t f_{B_i}$. Applying Claim~\ref{claim:sum-f-ineq} and then Claim~\ref{claim:conditional} twice, we get 
\begin{equation}\label{eq:parallel:psE}
\psE_B[f_B^r] \geq \sum_{i=1}^t\psE_B[f_{B_i}^r] = 
\sum_{i=1}^t\frac{1}{p_B}\psE[f_{B_i}^r] = 
\sum_{i=1}^t\frac{p_{B_i}}{p_B}\psE_{B_i}[f_{B_i}^r].
\end{equation}
In fact the inequality above is an equality by~Claim~\ref{claim:flow-basics}, part 2, but we do not need that here.
Comparing (\ref{eq:parallel:E}) and~(\ref{eq:parallel:psE}) term-by-term, and using the induction hypothesis, we get
$$\EE{B}{\cost(P_{B})^r} \leq \bell_\h(r) \psE_{B}[f_{B}^r].$$
This concludes the analysis of this case.
Now we assume that $B$ is a series composition of blocks $B_1,\dots, B_t$.
In this case, if $P$ visits $B$ then it visits all $B_i$; if it visits $B_i$, it visits $B$ and all other $B_j$. Thus, $p_B = p_{B_1} = \dots = p_{B_t}$. Also, $P_B$ is the concatenation of $P_{B_1}, \dots, P_{B_t}$.
Using the multinomial theorem, we get
\begin{align}
    \bE_B\pb{\cost(P_B)^r} \label{eq:E-cost-p}
        &= \bE_B\pb{\p{\sum_{i=1}^t \cost(P_{B_i})}^r} = \bE_B\pb{\sum_{\substack{\alpha_1, \ldots, \alpha_t \ge 0 \\ \alpha_1 + \cdots + \alpha_t = r}} \binom{r}{\alpha} \prod_{i=1}^t \cost(P_{B_i})^{\alpha_i}} \\
        &= \sum_{\substack{\alpha_1, \ldots, \alpha_t \ge 0 \\ \alpha_1 + \cdots + \alpha_t = r}} \binom{r}{\alpha}\bE_B \pb{\prod_{i=1}^t \cost(P \cap B_i)^{\alpha_i}}\notag.
    \end{align}
Observe that if $P$ enters \(B\), it necessarily visits $s_{B_1}, \dots, s_{B_t}$ and thus paths \(P_{B_1}, \ldots, P_{B_t}\) are mutually independent. We then have
\begin{align*}
    \bE_B\pb{\cost(P_B)^r} 
    &= \sum_{\substack{\alpha_1, \ldots, \alpha_t \ge 0 \\ \alpha_1 + \cdots + \alpha_t = r}} \binom{r}{\alpha}  \prod_{i=1}^t \bE_B\pb{\cost(P_{B_i})^{\alpha_i}}.
\end{align*}
By Claim~\ref{claim:conditional}, 
$\bE_B\pb{\cost(P_{B_i})^{\alpha_i}} =  \E{\cost(P_{B_i})^{\alpha_i}}/p_B = 
\E{\cost(P_{B_i})^{\alpha_i}}/p_{B_i}=
\bE_{B_i}\pb{\cost(P_{B_i})^{\alpha_i}}
$
and
$\psE_B[f_{B_i}^{\alpha_i}] = \psE[f_{B_i}^{\alpha_i}] / p_B= \psE[f_{B_i}^{\alpha_i}]/p_{B_i} = \psE_{B_i}[f_{B_i}^{\alpha_i}]$.
Using that $\sheight(B_i) \leq h-1$ and applying the induction hypothesis, we get
\begin{align*}
    \bE_B\pb{\cost(P_B)^r} &= \sum_{\substack{\alpha_1, \ldots, \alpha_t \ge 0 \\ \alpha_1 + \cdots + \alpha_t = r}} \binom{r}{\alpha}  \prod_{i=1}^t \bE_B\pb{\cost(P_{B_i})^{\alpha_i}}
    = \sum_{\substack{\alpha_1, \ldots, \alpha_t \ge 0 \\ \alpha_1 + \cdots + \alpha_t = r}} \binom{r}{\alpha}  \prod_{i=1}^t \bE_{B_i}\pb{\cost(P_{B_i})^{\alpha_i}} \\
    &\le \sum_{\substack{\alpha_1, \ldots, \alpha_t \ge 0 \\ \alpha_1 + \cdots + \alpha_t = r}} \binom{r}{\alpha}  \prod_{i=1}^t 
    \bell_{\h-1}(\alpha_i) \cdot \psE_{B_i} f_{B_i}^{\alpha_i} = \sum_{\substack{\alpha_1, \ldots, \alpha_t \ge 0 \\ \alpha_1 + \cdots + \alpha_t = r}} \binom{r}{\alpha}  \prod_{i=1}^t 
    \bell_{\h-1}(\alpha_i) \cdot \psE_{B} f_{B_i}^{\alpha_i}.
\end{align*}
Now, every $\alpha$ in the summation defines an unlabeled partition $\lambda$ of $n$: $\lambda$ is obtained by sorting all non-zero entries $\alpha_i$ of $\alpha$. Let us denote this $\alpha \to \sigma$. For example, $\alpha=(4,1,0,2,4,0,2, 2) \to \lambda=(4,4,2,2,2,1)$. Thus, to go over all $\alpha$, it is sufficient to go over all 
$\lambda \prt r$ and then all $\alpha$ such that $\alpha\to \lambda$. To do the latter, we go over all choices of distinct indices $j_1, \dots, j_{|\lambda|}$ and define $\alpha$ as follows: $\alpha_{j_i}=\lambda_i$ for $i\in [|\lambda|]$ and $\alpha_j = 0$ for all other $j$. However, if $\lambda_i = \lambda_{i'}$, then indices $j_1, \dots, j_{|\lambda|}$ and those with $j_i$ and $j_{i'}$ swapped define the same $\alpha$. It is easy to see that the above procedure defines every $\alpha$ exactly $\prod_{j=1}^r \cnt(j,\lambda)$ times. In the example above with $\lambda = (4,4,2,2,2,1)$, this procedure defines every $\alpha$ exactly $2!\cdot 3!$ times. Finally note that $\binom{r}{\alpha} = \binom{r}{\lambda}$.
Keeping this discussion in mind, we rewrite the upper bound on  $\bE_B\pb{\cost(P_B)^r}$ as follows.
\begin{align*}
         \bE_B\pb{\cost(P_B)^r} &
         \le \sum_{\lambda \prt r, ~|\lambda| \le t} \binom{r}{\lambda}  \frac{\sum_{\substack{j_1, \ldots, j_{|\lambda|} \in [t] \\ \text{distinct}}}\prod_{i=1}^{|\lambda|}
        \bell_{\h-1}(\lambda_i)\cdot \psE_{B} [f_{B_{j_i}}^{\lambda_i}] }{\prod_{j=1}^r \cnt(j,\lambda)}
        \\ 
        &\leq \sum_{\lambda \prt r, ~|\lambda| \le t} \binom{r}{\lambda} \cdot\frac{\prod_{i=1}^{|\lambda|} 
        \bell_{\h-1}(\lambda_i)}{\prod_{j=1}^r \cnt(j,\lambda)}  \sum_{j_1, \ldots, j_{|\lambda|} \in [t]} \prod_{i=1}^{|\lambda|}  \psE_{B}[ f_{B_{j_i}}^{\lambda_i}].
\end{align*}
Note that we removed the requirement that all $j_i$ are distinct in the last inequality (this is valid, since all terms are non-negative). Now we use Claim~\ref{claim:sum-f-ineq} and then the majorization inequality (see Lemma~\ref{lemma:ps_major}) to upper bound each term in the inner sum.
$$\sum_{j_1, \ldots, j_{|\lambda|} \in [t]} \prod_{i=1}^{|\lambda|}  \psE_{B}[ f_{B_{j_i}}^{\lambda_i}]
= \prod_{i=1}^{|\lambda|} \sum_{j\in [t]} \psE_{B}\Bigl[f_{B_{j}}^{\lambda_i}\Bigr]
=
\prod_{i=1}^{|\lambda|}  \psE_{B}\Bigl[\sum_{j\in [t]}f_{B_{j}}^{\lambda_i}\Bigr]
\leq 
\prod_{i=1}^{|\lambda|} \psE_{B}[ f_{B}^{\lambda_i}] \leq \psE[f_B^r].
$$
Using the recurrence relation for multidimensional Bell numbers from Lemma~\ref{claim:bell-recurrence}, we conclude that
$$\bE_B\pb{\cost(P_B)^r} \leq 
 \sum_{\lambda \prt r} \frac{\binom{r}{\lambda} \prod_{i=1}^{|\lambda|}
        \bell_{\h-1}(\lambda_i)}{\prod_{j=1}^r \cnt(j,\lambda)} \psE[f_B^r] = \bell_\h(r) \psE[f_B^r].
$$
\end{proof}

\begin{theorem}\label{thm:main:SP}
Algorithm~\ref{alg:rounding} gives an $(1+ \varepsilon) \abell_d(p) \eqdef(1+\varepsilon)\bell_d(p)^{1/p}$ approximation for the problem in series-parallel graphs of order $d$ in time polynomial in $n^{O(p)}$ and $1/\varepsilon$.
\end{theorem}
\begin{proof}
We apply Lemma~\ref{lem:cost-analysis} with $a_e =c_e(i)$ to every coordinate $i\in [\ell]$ and add up the obtained upper bounds on $\E{\bigl(\sum_{e\in P} c_e(i)\bigr)^p}$. We get that 
$$\E{\Bigl\|\sum_{e\in P} c_e\Bigr\|_p^p} \leq \bell_{\h}(p) \psE\Bigl[\sum_{i=1}^\ell \Bigl(\sum_e c_e(i) x_e\Bigr)^p\Bigr] \leq \bell_{\h}(p) \cdot \opt^p.$$
By Markov's inequality, $\Bigl\|\sum_{e\in P} c_e\Bigr\|_p^p \leq (1+\varepsilon) \bell_{\h}(p) \cdot \opt^p$ with probability at least $\varepsilon/(1+\varepsilon)$. By running the algorithm $1/\varepsilon$ times, we find a solution of cost at most $\Bigl\|\sum_{e\in P} c_e\Bigr\|_p \leq (1+\varepsilon)^{1/p} \bell_{\h}(p)^{1/p} \cdot \opt \leq (1+\varepsilon) \abell_\h(p) \cdot\opt$ with constant probability. (As is standard, we can run this procedure many times and make the failure probability exponentially small.)
\end{proof}

\section{Algorithms for \texorpdfstring{$\ell_p$}{l-p}-Shortest Path in Arbitrary Graphs}
\label{sec:general}
In this section, we describe an approximation algorithm for $\ell_p$-Shortest Path in arbitrary graphs. 

We note that there is a block-box reduction from problem in arbitrary graphs to that in series-parallel graphs, which is implicitly used by Li, Xu, and Zhang~\cite{li2023polylogarithmic} in their algorithm for $\ell_\infty$-Shortest Path (Robust $s$-$t$ Path). This reduction outputs a series-parallel graph with $O(n^{\log n})$ vertices, where $n$ is the number of vertices in the original graph. By using this reduction, we immediately get an $O(p \log^{1-1/p} n)$-approximation algorithm for general graphs with running time $n^{O(p\log n)}$. We describe how to get an approximation algorithm for general graphs with an improved running time and slightly improved approximation factor; namely we describe how to get a $O(c p\log^{1-1/p} n)$-approximation in time $m^{O( c e^{1/c} \log n)}$ for every $c\in(0,1/2)$.
We assume below that $\ell$ is at most polynomial in $n$; thus, we may assume $r\leq \roundup{\log_2 \ell} = O(\log n)$ (since all norms $\|\cdot\|_{r}$ with $r \geq \log_2 \ell$ are equivalent within a factor of 2).

\paragraph{Layered graphs and reduction from general graphs to layered graphs.}
We say that a DAG $G = (V, E)$ is an $s$-$t$ layered graph with $\Delta$ edge layers if 
$V$ is the disjoint union of vertex layers $V_0,V_1,\dots,V_{\Delta}$, $E$ is the disjoint union of layers $E_1,\dots, E_{\Delta}$, and each edge in \(E_i\) goes from \(V_{i-1}\) to \(V_i\).
Further, we require that $V_0=\{s\}$ and $V_\Delta = \{t\}$.

We transform an arbitrary graph $G$ with terminals $s$ and $t$ into a layered graph $\hat G=(V_{\hat{G}}, E_{\hat G})$ with $\Delta = n-1$ edge layers. We create vertex layers $V_0,V_1,\dots,V_{\Delta}$: $V_0= \{s\}$, $V_\Delta = \{t\}$, each $V_i$ is a disjoint copy of $V_G$. We connect $\hat u\in V_i$ with $\hat v\in V_{i+1}$ if there is an edge $(u,v)\in E_G$ between the corresponding vertices in $G$. The vector cost of $(\hat u, \hat v)$ equals that of $(u,v)$.
Additionally, we add \textit{padding} edges between copies of $t$ in adjacent layers and assign these edges cost $0$. 

For every $s$-$t$ path $P$ with at most $\Delta$ edges in $G$ there is a corresponding path $\hat P$ in $\hat G$ and vice versa ($P$ might not be a simple path); if path $P$ has $k < \Delta = n-1$ edges, then it contains $k$ non-padding edge and ends with $\Delta - k$ padding edges. Paths $P$ and $\hat P$ have the same vector costs.
Note that the $\ell_p$-Shortest Path $P^*$ between $s$ and $t$ is a simple path and thus contains at most $\Delta = n-1$ edges. Therefore, there is path $\hat P^*$ in $\hat G$ corresponding to $P^*$. An $\alpha$-approximation for $\hat P^*$ in $\hat G$ gives an $\alpha$-approximation for $P^*$ in $G$. 
This reduction shows that it is sufficient to consider layered graphs.

\paragraph{An algorithm for layered graphs.}
Assume that $G$ is a layered graph with $\Delta$ edge layers, source $s$, and sink $t$.
We use the SoS relaxation from Section~\ref{sec:sos-relaxation} for $G$. Let $a = \roundup{e^{1/c}}$.
We require that $\psE$ be a pseudo-expectation of degree $2d = 2(p+(a+1)\roundup{\log_{a+1} \Delta}) = \Theta(p + ce^{1/c}\log \Delta)$. 
For a set of edges $A$ of  size at most $d$, we define polynomial $h_A = \sum_{e\in A} x_a$ and conditional pseudo-expectation $\psE_A[\cdot]\eqdef \psE[\cdot\given h_A]$.
Given $A$ and a set of layer indices $I\subseteq [\Delta]$ so that $|A| + |I| \leq d$, we define a sampling procedure that samples an edge from each layer $E_i$ with $i\in I$ using pseudo-expectation $\psE_A$.
We assume that the two algorithms below have access to graph $G$ and pseudo-expectation $\psE$.

\begin{algorithm}[H]
\caption{Edge sampling procedure}\label{alg:sampling}
\begin{algorithmic}[1]
\State \textbf{Input:} a subset of layer indices $I\subseteq[\Delta]$ and a subset $A$ of edges.
\State \textbf{Output:} one edge from every layer $E_i$ with $i\in I$.
\State $R = \varnothing$
\ForAll{$i\in I$}
    \State Sample $e\in E_i$ with probability of choosing $e$ equal to $\psE_{A\cup R}[x_e]$
    \State $R = R\cup \{e\}$
\EndFor
\State \Return $R$
\end{algorithmic}
\end{algorithm}

 We say that we sample edges $e_1,\dots, e_k$ in layers $i_1,\dots, i_k$ conditioning on set $A$ to mean that we run Algorithm~\ref{alg:sampling} with parameters $I = \{i_1,\dots,i_k\}$ and $A$.

\begin{algorithm}[H]
\caption{Rounding algorithm for layered graphs}\label{alg:rounding-general}
\begin{algorithmic}[1]
\State \textbf{Input:} indices $y$ and $z$ of two edge layers ($1\leq y\leq z\leq \Delta$) and a subset of edges $A$.
\State \textbf{Output:} a path in $\hat G$ traversing layers $E_y$ to $E_z$.
\Function{FindPath}{$y, z, A$}
\If{$z - y  + 1 \leq a$}
    \State Sample edges $e_0, \dots, e_{z-y}$ in layers $E_y, E_{y+1},\dots, E_z$ conditioning on $A$. 
    \State\Return the path formed by $e_0,\dots, e_{z-y}$.
\EndIf
\State Let $m_i = y + \roundup{\frac{z-y}{a+1}\cdot i}$ for $i\in[a]$.
\State Sample edges $e_1,\dots, e_{a}$ in layers $m_1,\dots, m_a$ conditioning on $A$.
\State  Let $A' = A \cup \{e_{i}:i\in [a]\}$.\label{alg_line:random_a_p}
\State Let $m_0 = z-1$ and $m_{a+1} = y+1$.
\For{$i=0$ \textbf{to} $a$}
    \State $P_i = \mathrm{FindPath}(m_i+1, m_{i+1}-1, A')$ unless $m_i+1 > m_{i+1}-1$ then $P_i=\varnothing$
\EndFor
\State Let $P$ be the path formed by $P_0,e_1, P_1, e_2,\dots, e_a,P_a$.
\State \Return \(P\)\label{alg_line:rounding_sp_return_layered}
\EndFunction
\end{algorithmic}
\end{algorithm}

To solve the problem, we run the algorithm with $y=1$, $z=\Delta$, and $A=\varnothing$. We denote the path that the algorithm finds by $P$.
The analysis of this algorithm is quite similar to that of Algorithm~\ref{alg:rounding} for series-parallel graphs.
Instead of the series-parallel decomposition of $G$, we will consider the recursion tree whose nodes are invocations of FindPath. The height $h = h_\nu$ of a node $\nu$ is the length of the longest branch leaving $\nu$  (e.g., the height of leaves is 0).
Let us make a few observations and introduce some notation about the recursion:
\begin{itemize}
\item Each recursion node is uniquely determined by parameters $y=y_\nu$ and $z=z_\nu$ (which are not random). 
\item For each layer $E_i$, there is a non-random node $\nu$ where an edge $e\in E_i$ is sampled (this edge $e$ will be part of $P$).
\item For a node $\nu$, let $E_\nu = E_{x_\nu}\cup \dots \cup E_{y_\nu}$ and $P_{\nu} = P\cap E_\nu$. When FindPath returns from node $\nu$, it returns path $P_\nu$ that traverses all of the layers forming $E_\nu$. 
\item Let $A_\nu$ be the value of $A$ when we execute node $\nu$; $A_\nu$ is a random set of edges. It contains exactly the edges sampled at ancestor nodes of $\nu$.
 While $A_\nu$ is random, the set of layers hit by $A_\nu$ 
 is not random and depends only on $\nu$. 
 We say that a set of edges $A$ is $\nu$-admissible if $\Prob{A_\nu = A} > 0$.
\item For a $\nu$-admissible set $A$, define $\EE{A}{\cdot} \eqdef \E{\cdot \given A\subseteq P}$.
\end{itemize}
For a set of edges $A$, let $p_A \eqdef \Prob{A\subseteq P}$.
\begin{claim}
Assume that $A$ is $\nu$-admissible. Then
$$
p_A \eqdef \Prob{A\subseteq P} = \psE[h_A].
$$
\end{claim}
\begin{proof}
Let us say $A = \{e_1,\dots, e_k\}$ where edges $e_i$ are listed in the same order as they are sampled. Note that when edge $e_i$ is sampled, the pseudo-distribution is conditioned on $\{e_1,\dots,e_{i-1}\}$. We prove the desired identity by induction:
\begin{align*}
    \Prob{e_1,\dots,e_i\in P} &= \Prob{e_i\in P\given e_1,\dots, e_{i-1}\in P} \cdot \Prob{e_1,\dots, e_{i-1}\in P} \\
    =& \psE_{e_1,\dots, e_{i-1}}[x_{e_i}] \cdot\psE[h_{\{e_1,\dots, e_{i-1}\}}] = \psE[h_{\{e_1,\dots,e_i\}}].
\end{align*}
\end{proof}
Importantly, the claim is generally not true for arbitrary sets $A$ of size at most $a\roundup{\log_a \Delta}$. The claim shows that the order in which edges are examined in the for-loop of Algorithm~\ref{alg:sampling} does not affect the distribution of $R$; namely, the probability of sampling $R$ equals $\psE_A[\prod_{e\in R}x_e]$.

\begin{claim} Algorithms~\ref{alg:sampling} and \ref{alg:rounding-general} satisfy the following properties.
\begin{enumerate}
    \item The sampling procedure in Algorithm~\ref{alg:sampling} is well defined. Namely, $\psE[x_i] \geq 0$ and $\sum_{e\in E_i}\psE[x_i] = 1$ (see line 5 of the algorithm).
    \item Algorithm~\ref{alg:rounding-general} always returns an $s$-$t$ path in $G$.
\end{enumerate}
\end{claim}
\begin{proof}
1. First, $\psE[x_e]= \psE[x_e^2] \geq 0$. Let $g_i = \sum_{e\in E_i} x_e$. Then
$\psE_A[g_i] \cdot \psE[h_A]= \psE[g_i \cdot h_A]$.
By Claim \ref{claim:flow-basics}, item 1, $\psE[g_ih_A] = \psE[g_1 h_A] = \psE[h_A]$. We conclude that $\psE_A[g_i] = 1$.

2. It is straightforward that the algorithm will sample an edge from each layer $E_i$. We need to show that these edges form a path. Let us say the algorithm samples an edge $e= (u,v)$ in $E_i$ and $e'=(u',v')$ in $E_{i+1}$. We prove that $v=u'$. Let us see where edges $e$ and $e'$ are sampled. Note that paths $P_j$ found in FindPath are not adjacent to each other. Therefore, it is not possible that $e$ and $e'$ are sampled in different branches of the recursion. The only two possibilities are:
\begin{itemize}
\item $e$ and $e'$ are sampled at the same node of the recursion tree,
\item one of the edges is sampled at a node $\nu$ and the other at a descendant of $\nu$.
\end{itemize}
In the former case, one of the edges, say $e$ is sampled first by Algorithm~\ref{alg:sampling}. Then $e'$ is sampled using distribution $\psE_{A\cup R}[x_{e'}]$ with $e\in R$. In the latter case, one of the edges, say $e'$, is sampled using distribution $\psE_{A'}[x_{e'}]$ with $e\in A'$. That is, in either case $e'$ is sampled using distribution $\psE_S[x_{e'}]$ with $e\in S$.\footnote{It is possible that the roles of edges $e$ and $e'$ are reversed, but that case is completely analogous to the one we consider in the proof.}. Now, $\psE[x_e\given h_S] = 1$. From the flow conservation constraint at vertex $v$, we get $\psE_S[\sum_{e'\in \delta^+(v)} x_{e'}\given h_S] = 1$. That is, $e'\in \delta^+(v)$ with probability 1, as required.
\end{proof}

Now we will upper bound the approximation factor that our algorithm gives.
Consider a set of non-negative edge costs $a_e$. As in Section~\ref{sec:rounding}, define $\cost(P') = \sum_{e\in P'} a_e$ for a path $P'$ and $f_{\nu} = \sum_{e\in E_\nu} a_e x_e$. 
We prove an upper bound for $\E{\cost(P)^r}$ (cf. Lemma~\ref{lem:cost-analysis}). 
\begin{lemma}\label{lem:cost-analysis-layered}
Let $\psE$ be a feasible solution for the SoS relaxation of $\ell_p$-Shortest Path. Let $\nu$ be a node in the recursion tree of height $\h=h_\nu$. For every $\nu$-admissible set $A$ with $p_a > 0$, we have
$$
\EE{A}{\cost(P_\nu)^{r}} \leq \bell_{\h}(r) \psE_A[f_\nu^{r}]. 
$$
\end{lemma}
\begin{proof}
If $\h=0$, then all edges in $P_\nu$ are sampled at $\nu$ (see lines 5-6 of Algorithm~\ref{alg:rounding-general}). It follows that $\EE{\nu A}{\cost(P_\nu)^{r}} = \psE_A[f_\nu^{r}]$, as required. So we assume below that $\h \geq 1$. The lemma trivially holds when $r = 0$, so we assume that $r\geq 1$.

Suppose that $\nu$ has children $\nu_1,\dots, \nu_t$. Let  $f_i = f_{\nu_i}$ for $i\in [t]$. Additionally, let $f_0 = \sum_{i=1}^a\sum_{e\in E_{m_i}} a_e x_{e}$. 
Similarly, let $P_i = P_{\nu_i}$ and $\cost_0 = \sum_{i=1}^a a_{e_i}$.
We have
$$
\cost(P_\nu) = \cost_0 + \sum_{i=1}^t \cost(P_i)\qquad\text{and}\qquad f_\nu = f_0 + \sum_{i=1}^t f_i.
$$

We consider execution of node $\nu$ with parameter $A$. Let $A'$ be a random set as defined on \Cref{alg_line:random_a_p} of Algorithm~\ref{alg:rounding-general}. Note that $A'$ is $\nu_i$-admissible for every $i$.
By the induction hypothesis for all nodes $\nu_i$ and exponents $q\leq r$,
$$
\EE{A'}{\cost(P_{\nu_i})^{q}} \leq \bell_{\h-1}(q) \psE_{A'}[f_{\nu_i}^{q}]
\qquad
\text{and}
\qquad
\EE{A'}{\cost_0^{q}} = \psE_{A'}[f_0^{q}].
$$
Using a multinomial expansion of $\cost(P_\nu)^p$ analogous to \eqref{eq:E-cost-p} and using the same algebraic transformations and inequalities as in Lemma~\ref{lem:cost-analysis}, we get
\begin{equation}\label{eq:main-layered-ineq}
\EE{A'}{\cost(P_\nu)^{r}} \leq \bell_{\h}(r) \psE_{A'}[f_\nu^{r}].
\end{equation}
Finally, we observe that
\begin{equation}\label{eq:main-layered-sum-1}\EE{A}{\cost(P_\nu)^{r}} = \sum_{A'} \frac{p_{A'}}{p_A}\EE{A'}{\cost(P_\nu)^{r}}
\end{equation}
and 
\begin{equation}\label{eq:main-layered-sum-2}
\psE_{A}[f_\nu^{r}] = \sum_{A'} \frac{p_{A'}}{p_A}\psE_{A'}[f_\nu^{r}]
\end{equation}
where both summations are over all sets $A'$ that can be generated on \Cref{alg_line:random_a_p} of the algorithm (or, in other words, the summation is over $A' = A'' \setminus A$ where  $A''$ is $\nu_1$-admissible). We upper bound each term of \eqref{eq:main-layered-sum-1} using \eqref{eq:main-layered-ineq} and then simplify the obtained sum using \eqref{eq:main-layered-sum-2}. This concludes the proof.
\end{proof}

We observe that the recursion tree has height $\h = O(\log_a \Delta) = O(c\log{\Delta})$.
Similarly to Theorem~\ref{thm:main:SP}, we obtain the following theorem for layered graphs.
\begin{theorem}\label{thm:main:layered}
Algorithm~\ref{alg:rounding} gives an $O(cp \log^{1-1/p} \Delta )$ approximation for the $\ell_p$-Shortest Path problem in layered graphs with $\Delta$ layers in time $m^{O(p + c e^{1/c} \log \Delta)}$ for $c\in(0,1/2)$.
\end{theorem}
As a corollary, we get the following result for arbitrary graphs.
\begin{theorem}\label{thm:main:general}
There is an $O(cp \log^{1-1/p} n)$ approximation algorithm for the $\ell_p$-Shortest Path problem in arbitrary graphs that runs in time $m^{O(p + c e^{1/c} \log n)}$ for $c\in(0,1/2)$.
\end{theorem}

\section{Algorithms for \texorpdfstring{$\ell_p$}{l-p}-Group ATSP and \texorpdfstring{$\ell_p$}{l-p}-Group Steiner Tree}\label{sec:ATSP-and-Steiner-Tree}
In this section, we describe our algorithms for the \(\ell_p\)-Group ATSP and \(\ell_p\)-Group Steiner Tree problems.
First, we consider an auxiliary variant of the $\ell_p$-Group ATSP problem in layered graphs. We are given a layered graph $G$ and $k$ subsets of edges $R_1,\dots, R_k$. The goal is to find an $s$-$t$ path with minimum possible $\ell_p$-cost that visits  exactly one edge from every group. We  refer to this problem as Layered $\ell_p$-Group ATSP (the problem is arguably not very natural on its own, but we will later reduce Group ATSP to this problem).

\begin{theorem}\label{thm:ATSP:main}
There exists a pseudo-approximation algorithm for Layered $\ell_p$-Group ATSP that given an instance of the problem with $\Delta$ layers and a parameter $c \in (1/\log \Delta, 1/2)$, generates a random $s$-$t$ path $P$ such that
\begin{enumerate}
    \item $\E{\cost(P)^p}^{1/p} \leq O(c p\log^{1-1/p} \Delta) \cdot \opt$, where $\cost(P)$ is the $\ell_p$-cost of $P$ w.r.t. vector costs.
    \item For every group $R_i$, $\Prob{P \text{ visits at least one edge in } R_i} \geq \Omega(\frac{1}{c\log \Delta})$.
\end{enumerate}
The algorithm runs in time $m^{O(p + c e^{1/c} \log \Delta)}$.
\end{theorem}
\begin{proof}
For every group $R_i$, add the following constraints to the SoS relaxation: $h_{R_i} = 1$ and $h_{R_i}^2 = 1$. Then solve the relaxation and run Algorithm~\ref{alg:rounding-general} to generate $P$. The upper bounds on the cost and running time follow immediately from Theorem~\ref{thm:main:layered}. Now consider a group $R_i$, and let us prove that item 2 holds for it. To this end, define $a_e = 1$ for $e\in R_i$ and $a_e = 0$ otherwise. Let random variable $\xi$ be the number of edges in $R_i$ that $P$ visits. Since $P$ visits every edge $e$ with probability $\psE[x_e]$, $\E{\xi} = \psE[\sum_{e\in R_i} x_e] = \psE[h_{R_i}] =1$. By Lemma~\ref{lem:cost-analysis-layered} applied to costs $a_e$ (with $r=2$), $\E{\xi^2} \leq O(c\log \Delta \psE[h_{R_i}^2]) = O(c \log \Delta)$. By the Paley--Zygmund Inequality,
$$\Prob{\xi>0} \geq \frac{\E{\xi}^2}{\E{\xi^2}} \geq \frac{\Omega(1)}{c\log \Delta}.$$
\end{proof}

Now we present the reduction from $\ell_p$-Group ATSP in arbitrary graphs to Layered $\ell_p$-Group ATSP, and argue that it yields an approximation algorithm for the former.
\begin{theorem}
For every $c\in (0,1/2)$, there exists a $O(c^2 p  \log^{2-1/p} n \log k)$-approximation algorithm for $\ell_p$-Group ATSP, whose running time is $m^{O(p+ce^{1/c} \log n)}$. We assume that the number of groups $k$ is polynomial in $n$. 
\end{theorem}
\begin{proof}
First, we may assume that all groups $R_i$ are disjoint. If not and a vertex $u$ belongs to several groups, we can create a copy of $u$ for each group that contains it and connect all copies of $u$ by edges of cost 0. 

We guess one vertex in group $R_k$ that is visited by the optimal tour. Denote it by $s$; we will ensure that our tour will visit $s$. We transform $G$ into a layered graph $\hat G$ with $s=s$, $t=s$ and $\Delta = nk -1$ as described in Section~\ref{sec:general}. We create edge groups $R_i'$ for all vertex groups $R_1,\dots, R_{k-1}$ as follows: if $u\in R_i$ then we add edges between copies of $u$ in adjacent layers of $\hat G$ to $R_i'$. Let us call all edges in $\bigcup_i R_i'$ \textit{check-in} edges. 

Now consider an optimal tour. As we assumed, it visits vertex $s$; thus, we think of it as a non-simple path $P^*$ from $s$ to $s$. For each $i$, let $v_i$ be the first vertex in group $R_i$ on $P^*$. Note that vertices $v_i$ partition path $P^*$ into $k$ subpaths; we may assume that these subpaths are simple paths (if they are not, by shortcutting them, we can only decrease the cost). Therefore, $P^*$ has length at most $(n-1)k$.

We construct path $\hat P^*$ corresponding to $P^*$ in $\hat G$: we keep track of path $P^*$ and whenever it follows edge $e=(u,v)$, $\hat P^*$ follows the corresponding edge $\hat e = (\hat u, \hat v)$, except that when $P^*$ visits $v_i$ for the first time, $\hat P^*$ takes the check-in edge between copies of $v_i$. If necessary, we append padding edges (connecting copies of $t$) to the end of path $\hat P$.

Note that $\hat P$ has at most $(n-1)k$ regular edges (those that are not padding and check-in edges) and at most $k-1$ check-in edges; that is, it has at most $\Delta = nk-1$ non-padding edges.
Now we run Algorithm~\ref{alg:rounding-general}  $T = \Theta(c\log \Delta\log k)$ times. We get $T$ paths in $\hat G$. Let $P_1,\dots, P_T$ be the corresponding $s$-$s$ paths in $G$ and $P_{alg}$ be their concatenation. By Theorem~\ref{thm:ATSP:main}, item 1, $\E{\cost(P_{alg})^p}^{1/p}\leq O(c T p \log^{1-1/p} (nk))$.
Consider a group $R_i$. Theorem~\ref{thm:ATSP:main}, item 2, implies that each of the \(P_1, \ldots, P_T\) visits group $R_i$ with probability at least $\Omega({\frac{1}{c\log \Delta}})$. Thus, by choosing the constant in the big-O formula for \(T\) large enough, we can ensure that \(P_{alg}\) visits $R_i$ with probability at least (say) $1 - \frac{1}{3k}$. 
By the union bound, \(P_{alg}\) visits all of the groups $\{R_i\}_{i}$ with probability at least $2/3$.

After conditioning on this success event, we have
\begin{align*}
    \E{\cost(P_{alg})^p\given \text{ success}}^{1/p} &\le \frac{1}{\Pr(\text{success})^{1/p}} \cdot \E{\cost(P_{alg})^p}^{1/p} \\ 
    &\leq  (3/2)^{1/p} \cdot O(T c p \log^{1-1/p} \Delta) \cdot \opt 
    \\&= 
 O(c^2 p  \log^{2-1/p} n \log k) \cdot \opt.
\end{align*}

Running this algorithm sufficiently many times and outputting the best path, we get the desired result.
\end{proof}

We conclude this section by deriving an algorithm for the $\ell_p$-Group Steiner Tree problem. This is a simple corollary of the result above, and is obtained by extending a standard reduction from the Steiner tree problem to the ATSP problem, to the case of vector costs.
\begin{theorem}
    For every $c\in (0,1/2)$, there exists a $O(c^2 p  \log^{2-1/p} n \log k)$-approximation algorithm for
    $\ell_p$-Group Steiner Tree in undirected graphs running in time $m^{O(p + c e^{1/c} \log n)}$. We assume that the number of groups $k$ is polynomial in $n$. 
\end{theorem}
\begin{proof}
    We proceed by reducing the $\ell_p$-Group Steiner Tree problem to the $\ell_p$-Group ATSP problem as follows. Given an undirected graph $G=(V,E,c)$ with vector edge costs, we turn it into a directed graph $G'$ by replacing each undirected edge $\{u,v\}$ with two directed edges $(u,v)$ and $(v,u)$ of equal vector cost $c_{uv})$. Given an instance $\mathcal{I} = (G,R_1, \dots , R_k)$ of the $\ell_p$-Group Steiner tree problem in undirected graphs, we construct an instance $\mathcal{I}' = (G', R_1, \dots , R_k)$ of the $\ell_p$-Group ATSP problem on \(G'\) with the same groups. It is easy to see that any feasible tour $T'$ in $G'$ can be turned into a feasible tree in $G$ of no greater $\ell_p$-cost, by ignoring the directions of the edges in $T'$ and dropping edges forming cycles.
    Similarly, given a feasible Steiner tree $T$ in $G$ of cost $C$, an Euler tour of $T$ yields a valid group ATSP tour $T'$ with $\ell_p$-cost $2C$. Hence, this reduction preserves the cost of every solution within a factor of $2$, and the result follows.
\end{proof}

\section{Hardness results}\label{sec:all-hardness-results}
\subsection{Arbitrary costs}\label{sec:hardness-from-CVP}
In this section, we argue that the assumption that all edge costs are entry-wise non-negative is 
indeed necessary to obtain a good approximation result. We prove that the version of the $\ell_p$-Shortest Path problem in which weight vectors can have negative coordinates is as hard as the $\ell_p$ closest vector problem in lattices.

We now recall the basic definitions required for this argument. For a more detailed overview of the background, we refer the reader to the book of Micciancio and Goldwasser~\cite{micciancio2002complexity}. A lattice is the set of all integer linear combinations of a collection $\mathcal{L} = \{\mathbf{v}_1, \dots , \mathbf{v}_d\}$ of linearly independent vectors in $\R^n$. The set $\mathcal{L}$ is referred to as the \emph{lattice basis} and is often represented as a matrix
\[
    B = \begin{pmatrix}
            | & &|\\
            \mathbf{v}_1& \cdots & \mathbf{v}_d\\
            | & &|
        \end{pmatrix}.
\]

The closest vector problem is then defined as follows.
\begin{definition}[$\ell_p$-Closest Vector Problem]
    The $\ell_p$-closest vector problem ($\ell_p$-CVP) is defined as follows: given a lattice basis $B \in \Z^{n\times d}$ and a target vector $\mathbf{t} \in \Z^n$, find a vector $\mathbf{x} \in \Z^d$ minimizing $\norm{B\mathbf{x} - \mathbf{t}}_p$.
\end{definition}

This problem has a long history and several hardness results have been proved about it. We will use hardness results by Dinur, Kindler, and Safra~\cite{dinur2003approximating} (for $p < \infty$) and  Dinur~\cite{dinur2002approximating} (for $p=\infty$).
\begin{theorem}[Paraphrasing~\cite{dinur2003approximating} and Dinur~\cite{dinur2002approximating}]\label{thm:known-CVP-hardness}
    For every $p \in[1,\infty]$, it is \textnormal{NP}-hard to approximate $\ell_p$-CVP to within a factor of $d^{\Omega_p(1 / \log \log d)}$.
\end{theorem}

Now we show that all hardness of approximation results for $\ell_p$-CVP, including Theorem~\ref{thm:known-CVP-hardness}, directly apply to the $\ell_p$-Shortest Path problem with negative coordinates allowed.
Combining it with \Cref{thm:known-CVP-hardness}, we readily obtain
\Cref{thm:intro-hardness-from-CVP}.

\begin{theorem}\label{thm:CVP-reduction}
    There is a polynomial-time approximation-preserving reduction from the $\ell_p$-CVP problem to the $\ell_p$-Shortest Path problem with negative coordinates allowed.
\end{theorem}

\begin{figure}
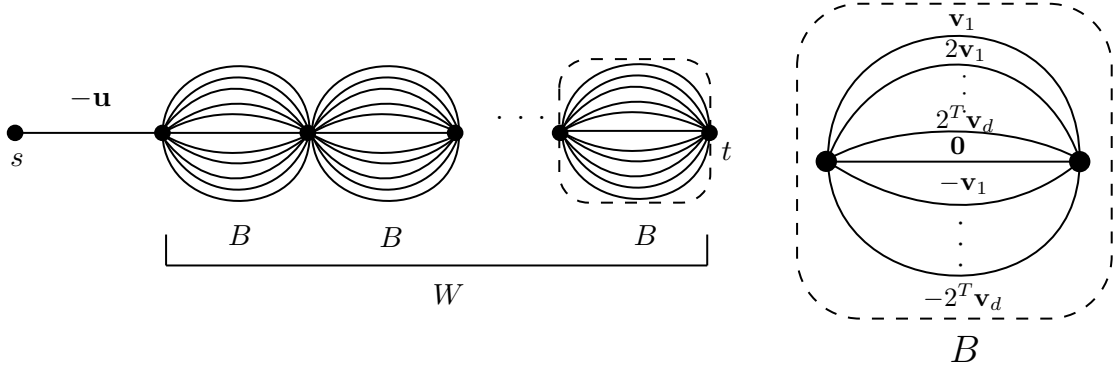

    \centering
    \include{img/CVP_reduction}
    \caption{The structure of the reduction for the proof of Theorem~\ref{thm:CVP-reduction}. On the left, we illustrate the structure of the $\ell_p$-Shortest Path instance graph $G$ constructed. On the right we show the $B$-block used to construct $G$.}
    \label{fig:CVP-reduction}
\end{figure}
\begin{proof}
    Given an instance of the closest vector problem defined by a lattice basis
    \[
        B = \begin{pmatrix}
            | & &|\\
            \mathbf{v}_1& \cdots & \mathbf{v}_d\\
            | & &|
        \end{pmatrix}
    \]
    and a target vector $\mathbf{u}$, we construct an instance $G=(V,E,w)$ of the $\ell_p$-Shortest Path problem as follows (see also Figure~\ref{fig:CVP-reduction}). 
    We will use two integers $T$ and $W$ to describe the size of the instance, as explained below. The vertex set $V$ consists of $W+2$ vertices $v_0, \dots , v_{W+1}$. Each vertex $v_i$ for $1 \le i \le W$ is connected to vertex $v_{i+1}$ by $2d(T+1)$ edges with vector costs equal to $\{2^j\mathbf{v}_k \given 1 \le k \le d, 0 \le j \le T\}$ and $\{-2^j\mathbf{v}_k \given 1 \le k \le d, 0 \le j \le T\}$ as well as an extra edge of cost $0$. Vertex $v_0$ is connected to $v_1$ with a single edge with weight $-\mathbf{u}$.
    We set the source vertex to be \(s = v_0\) and the destination vertex to be \(t = v_{W+1}\).

    We now show that for some suitable choice of $T, W =\text{poly}(d,\log M, \log n)$, where $M$ is an upper bound on the magnitude of the numbers in the CVP instance, this reduction simultaneously ensures that: (1) every feasible solution to the $\ell_p$-Shortest Path instance can be converted in polynomial time to a solution to the CVP of equal objective value, and (2) the optimal value of the $\ell_p$-Shortest Path instance is the same as the optimal value of the CVP instance. These two facts directly show the reduction is approximation-preserving, as in the statement of the theorem.

    \begin{claim}\label{claim:first-claim-in-CVP-proof}
        Every solution to the $\ell_p$-Shortest Path instance can be efficiently converted to a solution of the CVP instance with the same objective value.
    \end{claim}
    \begin{proof}[Proof of Claim~\ref{claim:first-claim-in-CVP-proof}]
        By construction, every s-t path \(P\) in the $\ell_p$-Shortest Path instance is a collection of $W+1$ edges, the first one of which is the edge $(v_0,v_1)$ with cost $-\mathbf{u}$ and all the other edges $i$ connect $v_i$ to $v_{i+1}$ with cost equal to a vector of the form $g_i \cdot 2^{j_i} \mathbf{v}_{k_i}$ where $g_i\in\{\pm 1\}$, $0 \le j_i \le T$, and $0\leq k_i\leq d$; here, to simplify notation, we denote $\mathbf{v}_0 = 0$ ($0$ is not a lattice basis vector). The $\ell_p$-cost of \(P\) is then
        \begin{equation}
           \cost_{\ell_p}(P) = \norm{\sum_{i=1}^W g_i \cdot 2^{j_i }\mathbf{v}_{k_i}- \mathbf{u}}_p.
        \end{equation}
        We now observe that the CVP solution $\mathbf{x}\in \Z^d$ given by
        \begin{equation}\label{eq:x-in-terms-of-path}
            \mathbf{x} = \begin{pmatrix}
                \sum_{i:k_i=1} g_i \cdot 2^{j_i}\\
                \vdots\\
                \sum_{i:k_i=d} g_i \cdot 2^{j_i}
            \end{pmatrix}
        \end{equation}
        \noindent
        has CVP objective value $\norm{B\mathbf{x} - \mathbf{u}}_p$ equal to $\cost_{\ell_p}(P)$. This completes the proof of the first claim.
        \end{proof}
        
        We now address the second claim.
        \begin{claim}\label{claim:second-claim-in-CVP-proof}
            Assume that $T, W =\emph{\text{poly}}(d,\log M, \log n)$ are sufficiently large. For the optimal solution of the CVP problem, there is a corresponding solution of the $\ell_p$-Shortest Path instance of the same cost.
        \end{claim}
        \begin{proof}[Proof of Claim~\ref{claim:second-claim-in-CVP-proof}]
            Consider the optimal solution $\mathbf{x}^*$ to the CVP instance. 
            The corresponding solution of the \(\ell_p\)-Shortest Path instance will be given by
            the path corresponding to the binary expansion of the entries of \(\x^*\). Namely,
            for each \(1 \le k \le d\) consider the binary representation of \(\x^*_k\), as
            \[\x^*_k = g_k \sum_{i=0}^{\lceil \log |\x^*_k| \rceil} j_k^{(i)} 2^{i}\]
            where \(g_k \in \{\pm 1\}\), and \(j_k^{(i)} \in \{0, 1\}\) for each \(0 \le i \le \lceil \log_2 |\x^*_k| \rceil\).
            First we assume that \(T\) and \(W\) are large enough (we will fix them later), and we construct the path \(P\) corresponding to this solution edge-by-edge,
            with the first edge being \(v_0 \to v_1\), of cost \(-\mathbf{u}\).
            Then the path follows the binary representation:
            for each \(1 \le k \le d\)
            and \(0 \le i \le \lceil \log_2 |\x_k^*| \rceil\) such that \(j_k^{(i)} = 1\),
            we append the edge of cost \(g_k \cdot 2^{i} \cdot \mathbf{v}_k\) to the end of \(P\). The rest of the path is formed by edges of cost zero.
            
            Now by construction, we have
            \[\cost_{\ell_p}(P) = \norm{\sum_{k = 1}^d \p{g_k \sum_{i=0}^{\lceil \log_2 |\x_k^*| \rceil} j_k^{(i)} 2^{i}} \cdot \mathbf{v}_k - \mathbf{u}}_p = \norm{\sum_{k=1}^d \x^*_k \cdot \mathbf{v}_k - \mathbf{u}}_p = \norm{B \x^* - \mathbf{u}}_p\]

            It remains to find \(T\) and \(W\) large enough for the construction of \(P\) above to work.
            Observe that it suffices for \(T \ge \log \|\x^*\|_\infty \)
            and \(W \ge \sum_{k=1}^d |\{0 \le i \le \lceil \log_2 |\x_k^*| \rceil \colon j_k^{(i)} = 1\}|\).
            To this end, we show that
            \[\log \|\x^*\|_\infty \le O(d \log nd + d \log M)\]
            and so taking \(T = O(d \log nd + d \log M)\)
            and \(W = d \cdot T\) suffices.
            
            Let $\mathbf{w}^* = B \mathbf{x}^*$, so that \(
                \opt_{\text{CVP}} = \norm{\mathbf{w}^* - \mathbf{u}}_p\). Note that since $\mathbf{x} = 0$ is always a feasible solution,
            \begin{equation}\label{eq:bound-on-opt-CVP}
                \opt_{\text{CVP}} \leq \norm{\mathbf{u}}_p.
            \end{equation}
            We then have
            \[
                \norm{\mathbf{w}^*}_\infty \leq \norm{\mathbf{u}}_\infty + \norm{\mathbf{w}^* -\mathbf{u}}_\infty \leq \norm{\mathbf{u}}_\infty + \opt_{\text{CVP}}\leq 2\norm{\mathbf{u}}_p  \leq 2{n}M.
            \]
            
            Note that, since $B$ is full-rank, $\mathbf{x}^*$ is the only solution to $B\mathbf{x} = \mathbf{w}^*$ and hence $\mathbf{x}^* = B^{\dagger} \mathbf{w}$ where $B^\dagger = (B^{\top}B)^{-1} B^{\top}$ is the Moore–Penrose inverse of $B$. Recall the well-known formula for the inverse
            \begin{equation}
                (B^{\top}B)^{-1}_{ij} = \frac{(-1)^{i+j}}{\operatorname{det}(B^{\top}B)} M_{ij},
            \end{equation}
            where $M_{ij}$ is the $(i,j)$ cofactor of $B^{\top}B$, that is, $M_{ij} = \operatorname{det}(B^\top B)^{(-i,-j)}$, where $(B^\top B)^{(-i,-j)}$ is obtained from $(B^\top B)$ by removing the $i$-{th} row and the $j$-{th} column. 
            
            Now observe that because \(B\) only has integer entries, \(\det(B^\top B)\) is also an integer. As \(B\) is full-rank, we then get \(\det(B^\top B) \ge 1\). By the permutation formula for the determinant, we have:
            \begin{align*}
               |(B^\top B)^{-1}_{ij}| \leq |M_{ij}| &= \left|\sum_{\sigma \in S_{d-1}} \sgn(\sigma) \prod_{k=1}^{d-1} (B^\top B)_{k\sigma(k)}^{(-i,-j)}\right|\leq (d-1)! \p{\max_{i,j\in [d]} \left|(B^\top B)_{ij}\right|}^{d-1} \\
               &\leq (d-1)! (M^2 \cdot n)^{d-1}
            \end{align*}
            where $S_{d-1}$ is the group of permutations of $d-1$ elements and $\sgn(\sigma)$ is the sign of the permutation $\sigma$. 
            Hence we have
            \begin{equation}
               |B^\dagger_{ij}| \leq d \cdot M \max_{i,j\in[d]} |(B^\top B)^{-1}_{ij}| \leq d!\cdot M^{2d-1} \cdot n^{d-1}.
            \end{equation}
            and
            \begin{equation}
                \norm{\mathbf{x}^*}_\infty \leq \norm{\mathbf{w}^*}_{\infty} \cdot n^d \cdot d!\cdot M^{2d-1}  \leq 2n \cdot (nd)^d \cdot M^{2d}.
            \end{equation}
            Thus,
            \[
                \log_2 \norm{\mathbf{x}^*}_\infty \leq 1 + \log_2 n + d\log_2 nd + (2d+1)\log_2 M.
            \]
            as desired.


        \end{proof}

        This concludes the proof of Theorem~\ref{thm:CVP-reduction}.

\end{proof}

\subsection{Reduction from Congestion Minimization}\label{sec:hardness-from-congestion-minimization}
In this section, we show $\Omega(\log n / \log \log n)$-hardness of approximation for the $\ell_\infty$-shortest path problem. We do so by providing an approximation-preserving reduction from the \emph{Directed Congestion Minimization} problem, which is defined as follows. Given a directed unweighted graph $G=(V,E)$ and a collection of source-sink pairs $\{(s_1,t_1),(s_2,t_2) \dots , (s_k,t_k)\}$, we need to find $s_1$-$t_1$ path $P_1$, $s_2$-$t_2$ path $P_2$,\dots, $s_k$-$t_k$ path $P_k$. Define the congestion $c$ of this family of paths as the maximum number of times a single edge is used $c= \max_e|\{i:e\in P_i\}|$. Then the goal is to find a feasible solution that minimizes $c$. Chuzhoy and Khanna~\cite{chuzhoy2006hardness} showed that this problem is hard to approximate within an $\Omega(\log n/\log\log n)$-factor unless $\mathrm{NP} \subseteq \mathrm{ZPTIME}(n^{\log \log n})$.


\begin{theorem}\label{thm:minimum-congestion-reduction}
There is a polynomial-time approximation-preserving reduction from Directed Congestion Minimization to $\ell_p$-Shortest Path.
\end{theorem}
This, together with the result of Chuzhoy and Khanna~\cite{chuzhoy2006hardness}, immediately gives us the following hardness result for $\ell_\infty$-Shortest Path.
\begin{corollary}
    The $\ell_\infty$-Shortest Path problem is hard to approximate within a factor of $\Omega(\log n/ \log \log n)$ unless $\mathrm{NP} \subseteq \mathrm{ZPTIME}(n^{\log \log n})$.
\end{corollary}
  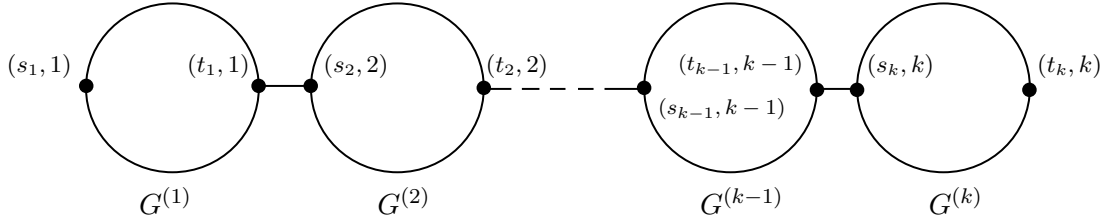
\begin{figure}[t]
        \centering

\tikzset{every picture/.style={line width=0.75pt}} 

\begin{tikzpicture}[x=0.75pt,y=0.75pt,yscale=-1,xscale=1]

\draw   (81,101.47) .. controls (81,78.01) and (100.3,59) .. (124.12,59) .. controls (147.93,59) and (167.23,78.01) .. (167.23,101.47) .. controls (167.23,124.92) and (147.93,143.93) .. (124.12,143.93) .. controls (100.3,143.93) and (81,124.92) .. (81,101.47) -- cycle ;
\draw   (193.23,101.47) .. controls (193.23,78.01) and (212.54,59) .. (236.35,59) .. controls (260.16,59) and (279.47,78.01) .. (279.47,101.47) .. controls (279.47,124.92) and (260.16,143.93) .. (236.35,143.93) .. controls (212.54,143.93) and (193.23,124.92) .. (193.23,101.47) -- cycle ;
\draw   (359.47,101.47) .. controls (359.47,78.01) and (378.77,59) .. (402.58,59) .. controls (426.4,59) and (445.7,78.01) .. (445.7,101.47) .. controls (445.7,124.92) and (426.4,143.93) .. (402.58,143.93) .. controls (378.77,143.93) and (359.47,124.92) .. (359.47,101.47) -- cycle ;
\draw   (465.7,101.47) .. controls (465.7,78.01) and (485,59) .. (508.82,59) .. controls (532.63,59) and (551.93,78.01) .. (551.93,101.47) .. controls (551.93,124.92) and (532.63,143.93) .. (508.82,143.93) .. controls (485,143.93) and (465.7,124.92) .. (465.7,101.47) -- cycle ;
\draw  [fill={rgb, 255:red, 0; green, 0; blue, 0 }  ,fill opacity=1 ] (78,100.47) .. controls (78,98.75) and (79.4,97.35) .. (81.12,97.35) .. controls (82.84,97.35) and (84.23,98.75) .. (84.23,100.47) .. controls (84.23,102.19) and (82.84,103.58) .. (81.12,103.58) .. controls (79.4,103.58) and (78,102.19) .. (78,100.47) -- cycle ;
\draw  [fill={rgb, 255:red, 0; green, 0; blue, 0 }  ,fill opacity=1 ] (164,100.47) .. controls (164,98.75) and (165.4,97.35) .. (167.12,97.35) .. controls (168.84,97.35) and (170.23,98.75) .. (170.23,100.47) .. controls (170.23,102.19) and (168.84,103.58) .. (167.12,103.58) .. controls (165.4,103.58) and (164,102.19) .. (164,100.47) -- cycle ;
\draw  [fill={rgb, 255:red, 0; green, 0; blue, 0 }  ,fill opacity=1 ] (190,100.47) .. controls (190,98.75) and (191.4,97.35) .. (193.12,97.35) .. controls (194.84,97.35) and (196.23,98.75) .. (196.23,100.47) .. controls (196.23,102.19) and (194.84,103.58) .. (193.12,103.58) .. controls (191.4,103.58) and (190,102.19) .. (190,100.47) -- cycle ;
\draw  [fill={rgb, 255:red, 0; green, 0; blue, 0 }  ,fill opacity=1 ] (276.35,101.47) .. controls (276.35,99.75) and (277.75,98.35) .. (279.47,98.35) .. controls (281.19,98.35) and (282.58,99.75) .. (282.58,101.47) .. controls (282.58,103.19) and (281.19,104.58) .. (279.47,104.58) .. controls (277.75,104.58) and (276.35,103.19) .. (276.35,101.47) -- cycle ;
\draw  [fill={rgb, 255:red, 0; green, 0; blue, 0 }  ,fill opacity=1 ] (356.35,101.47) .. controls (356.35,99.75) and (357.75,98.35) .. (359.47,98.35) .. controls (361.19,98.35) and (362.58,99.75) .. (362.58,101.47) .. controls (362.58,103.19) and (361.19,104.58) .. (359.47,104.58) .. controls (357.75,104.58) and (356.35,103.19) .. (356.35,101.47) -- cycle ;
\draw  [fill={rgb, 255:red, 0; green, 0; blue, 0 }  ,fill opacity=1 ] (442.58,102.47) .. controls (442.58,100.75) and (443.98,99.35) .. (445.7,99.35) .. controls (447.42,99.35) and (448.82,100.75) .. (448.82,102.47) .. controls (448.82,104.19) and (447.42,105.58) .. (445.7,105.58) .. controls (443.98,105.58) and (442.58,104.19) .. (442.58,102.47) -- cycle ;
\draw  [fill={rgb, 255:red, 0; green, 0; blue, 0 }  ,fill opacity=1 ] (462.58,102.47) .. controls (462.58,100.75) and (463.98,99.35) .. (465.7,99.35) .. controls (467.42,99.35) and (468.82,100.75) .. (468.82,102.47) .. controls (468.82,104.19) and (467.42,105.58) .. (465.7,105.58) .. controls (463.98,105.58) and (462.58,104.19) .. (462.58,102.47) -- cycle ;
\draw  [fill={rgb, 255:red, 0; green, 0; blue, 0 }  ,fill opacity=1 ] (548.58,102.47) .. controls (548.58,100.75) and (549.98,99.35) .. (551.7,99.35) .. controls (553.42,99.35) and (554.82,100.75) .. (554.82,102.47) .. controls (554.82,104.19) and (553.42,105.58) .. (551.7,105.58) .. controls (549.98,105.58) and (548.58,104.19) .. (548.58,102.47) -- cycle ;
\draw    (167.12,100.5) -- (190,100.5) ;
\draw    (445.7,102) -- (465.7,102) ;
\draw    (279.47,102) -- (293.23,102) ;
\draw    (339.7,102) -- (359.47,102) ;
\draw    (298.47,102) -- (307.23,102) ;
\draw    (312.47,102) -- (320.23,102) ;
\draw    (326.47,102) -- (334.23,102) ;

\draw (40,83) node [anchor=north west][inner sep=0.75pt]  [font=\footnotesize]  {$( s_{1} ,1)$};
\draw (130,83) node [anchor=north west][inner sep=0.75pt]  [font=\footnotesize]  {$( t_{1} ,1)$};
\draw (198,83) node [anchor=north west][inner sep=0.75pt]  [font=\footnotesize]  {$( s_{2} ,2)$};
\draw (279,83) node [anchor=north west][inner sep=0.75pt]  [font=\footnotesize]  {$( t_{2} ,2)$};
\draw (365,105) node [anchor=north west][inner sep=0.75pt]  [font=\scriptsize]  {$( s_{k-1} ,k-1)$};
\draw (375,83) node [anchor=north west][inner sep=0.75pt]  [font=\scriptsize]  {$( t_{k-1} ,k-1)$};
\draw (469,83) node [anchor=north west][inner sep=0.75pt]  [font=\footnotesize]  {$( s_{k} ,k)$};
\draw (554,83) node [anchor=north west][inner sep=0.75pt]  [font=\footnotesize]  {$( t_{k} ,k)$};
\draw (106,150) node [anchor=north west][inner sep=0.75pt]    {$G^{( 1)}$};
\draw (225,150) node [anchor=north west][inner sep=0.75pt]    {$G^{( 2)}$};
\draw (386,150) node [anchor=north west][inner sep=0.75pt]    {$G^{( k-1)}$};
\draw (500,150) node [anchor=north west][inner sep=0.75pt]    {$G^{( k)}$};

\end{tikzpicture}
        \caption{A diagram of the construction for the proof of Theorem~\ref{thm:minimum-congestion-reduction}.}
        \label{fig:minimum-congestion-reduction}
    \end{figure}
\begin{proof}[Proof of Theorem~\ref{thm:minimum-congestion-reduction}]
Our reduction works as follows. Given an instance $G=(V,E)$ of the congestion minimization problem, we construct an instance $G'=(V',E',w)$ of the $\ell_\infty$-Shortest Path problem defined as follows
    \[
        V'\eqdef\{(v,i) \given v \in V, i\in [k]\}, 
    \]
    \[
        E'\eqdef \{((u,i),(v,i))\in E \given (u,v)\in E, i\in [k]\} \cup\{ ((t_i,i),(s_{i+1},{i+1}))\given i \in [k-1]\},
    \]
    and $c : E' \to \R^E$ (that is, the vector costs in this instance are $m$-dimensional, where $m = |E|$, and the coordinates are indexed by $e\in E$) given by
    \[
        c(((u,i),(v,i)))_{(u',v')} = \begin{cases}
            1, & \text{if }  (u,v) = (u',v'),\\
            0, & \text{otherwise.}
        \end{cases}
    \]
    The costs of edges $((t_i,i),(s_{i+1},{i+1}))$ is 0.
    
    In other words, we construct $k$ copies $G^{(1)}, \dots , G^{(k)}$ of the Directed Congestion Minimization instance $G=(V,E)$ and stitch them together by connecting terminal $t_i$ in the $i$-{th} copy to terminal $s_{i+1}$ in the $(i+1)$-{th} copy. Copies $((u,1),(v,1)),\dots,((u,k), (v,k))$ of the same edge $(u,v)$ have the same cost. Note that the edges of the form $((t_i,i),(s_i,i+1))$ have cost zero, though one could easily change the construction to enforce that every cost vector is non-zero by simply identifying the two endpoints with each other.

    We now argue that there is a one-to-one correspondence between routings in $G$ with congestion $c$ and paths connecting $(s_1,1)$ and $(t_k,k)$ in $G'$ with $\ell_\infty$ cost $c$. Consider a solution $P_1,\dots,P_k$ for the Directed Congestion Minimization problem with congestion $c$. We construct a path $P$ from $(s_1,1)$ to $(t_k,k)$ in $G'$ by taking a copy of $P_1$ in $G^{(1)}$, a copy of $P_2$ in $G^{(2)}$, \dots, a copy of $P_k$ in $G^{(k)}$, concatenating all of them and all zero-cost edges between graphs $G^{(i)}$ and $G^{(i+1)}$. By construction, this path has $\ell_\infty$-cost equal to $c$. Similarly, every path in $G'$ connecting $(s_1,1)$ to $(t_k,k)$ with $\ell_\infty$-cost $c$ can be decomposed into $k$ paths $P_1, \dots, P_k$ (plus zero-cost edges) all of which are simultaneously routable in $G$ with congestion $c$.

    We have obtained an approximation-preserving reduction from Directed Congestion Minimization to $\ell_p$-Shortest Path.
\end{proof}

We remark that we can consider an $\ell_p$-version of the Directed Congestion Minimization problem in which congestion is measured as the $\ell_p$-norm of the vector of congestions on the edges of $G$ (the \emph{maximum} congestion used above corresponds to $\ell_\infty$). A similar analysis can be used to reduce this version of Congestion Minimization to the $\ell_p$-Shortest Path problem, showing that $\ell_p$-Shortest Path is at least as hard to approximate as $\ell_p$-Directed Congestion Minimization for every $p$. However, we are not aware of any known hardness results for this more general version of Congestion Minimization, so this result would not directly yield any lower bounds for  $\ell_p$-Shortest Path.

\subsection{Tightness of the analysis in Section~\ref{sec:rounding}}\label{sec:tightness-of-analysis}
In this section, we argue that our analysis provided in Lemma~\ref{lem:cost-analysis} (see Section \ref{sec:rounding}) is tight. 
We start with considering a simple construction showing that Lemma~\ref{lem:cost-analysis} is tight for $h=1$. We then refine this construction and show that the result of the lemma is also tight for arbitrary values of $h$. Initially, we analyze the performance of the algorithm on a feasible but suboptimal SoS solution. We then show that the algorithm performs the same on an optimal SoS solution.

For brevity, we will refer to the $\ell_p$-cost of a path raised to the power of $p$ as the $\ell_p^p$-cost and the SoS-objective as the $\ell_p^p$-objective.

\paragraph{Base case} Let $H$ be an undirected graph consisting of only two vertices with two edges between them $e_1$ and $e_2$ with scalar cost $c(e_1)=1$ and $c(e_2) = 0$. Let $G^{(1)}_N$ be the order 1 series-parallel graph obtained by composing together $N$ copies of $H$ in series (see Figure~\ref{fig:first_tightness_proof}). 
\begin{figure}[H]
    \centering
    \tikzset{every picture/.style={line width=0.75pt}} 

\begin{tikzpicture}[x=0.75pt,y=0.75pt,yscale=-1,xscale=1]

\draw  [fill={rgb, 255:red, 0; green, 0; blue, 0 }  ,fill opacity=1 ] (139.08,115.17) .. controls (139.08,112.56) and (141.2,110.43) .. (143.82,110.43) .. controls (146.44,110.43) and (148.56,112.56) .. (148.56,115.17) .. controls (148.56,117.79) and (146.44,119.92) .. (143.82,119.92) .. controls (141.2,119.92) and (139.08,117.79) .. (139.08,115.17) -- cycle ;
\draw  [fill={rgb, 255:red, 0; green, 0; blue, 0 }  ,fill opacity=1 ] (233.49,115.17) .. controls (233.49,112.56) and (235.62,110.43) .. (238.24,110.43) .. controls (240.86,110.43) and (242.98,112.56) .. (242.98,115.17) .. controls (242.98,117.79) and (240.86,119.92) .. (238.24,119.92) .. controls (235.62,119.92) and (233.49,117.79) .. (233.49,115.17) -- cycle ;
\draw  [fill={rgb, 255:red, 0; green, 0; blue, 0 }  ,fill opacity=1 ] (329.22,115.17) .. controls (329.22,112.56) and (331.35,110.43) .. (333.96,110.43) .. controls (336.58,110.43) and (338.71,112.56) .. (338.71,115.17) .. controls (338.71,117.79) and (336.58,119.92) .. (333.96,119.92) .. controls (331.35,119.92) and (329.22,117.79) .. (329.22,115.17) -- cycle ;
\draw  [fill={rgb, 255:red, 0; green, 0; blue, 0 }  ,fill opacity=1 ] (494.45,115.17) .. controls (494.45,112.56) and (496.58,110.43) .. (499.2,110.43) .. controls (501.82,110.43) and (503.94,112.56) .. (503.94,115.17) .. controls (503.94,117.79) and (501.82,119.92) .. (499.2,119.92) .. controls (496.58,119.92) and (494.45,117.79) .. (494.45,115.17) -- cycle ;
\draw    (143.82,116.06) .. controls (144.53,129.73) and (166.86,143.54) .. (189.81,144.25) .. controls (212.77,144.96) and (239.28,131.79) .. (239.55,116.06) ;
\draw  [fill={rgb, 255:red, 0; green, 0; blue, 0 }  ,fill opacity=1 ] (397.41,115.17) .. controls (397.41,112.56) and (399.54,110.43) .. (402.16,110.43) .. controls (404.78,110.43) and (406.9,112.56) .. (406.9,115.17) .. controls (406.9,117.79) and (404.78,119.92) .. (402.16,119.92) .. controls (399.54,119.92) and (397.41,117.79) .. (397.41,115.17) -- cycle ;
\draw    (146.18,202.02) -- (497.72,202.02) ;
\draw    (146.18,202.02) -- (146.18,181.75) ;
\draw    (497.72,202.02) -- (497.72,181.75) ;
\draw    (143.82,116.06) .. controls (144.11,98.55) and (168.06,85.63) .. (192.16,85.66) .. controls (216.26,85.69) and (240.2,102.07) .. (239.55,116.06) ;
\draw    (237.56,117.24) .. controls (238.27,130.9) and (260.6,144.71) .. (283.56,145.42) .. controls (306.51,146.14) and (333.02,132.96) .. (333.29,117.24) ;
\draw    (237.56,117.24) .. controls (237.86,99.72) and (261.8,86.81) .. (285.9,86.83) .. controls (310,86.86) and (333.95,103.24) .. (333.29,117.24) ;
\draw    (402.78,116.06) .. controls (403.5,129.73) and (425.82,143.54) .. (448.78,144.25) .. controls (471.74,144.96) and (498.25,131.79) .. (498.51,116.06) ;
\draw    (402.78,116.06) .. controls (403.08,98.55) and (427.02,85.63) .. (451.12,85.66) .. controls (475.23,85.69) and (499.17,102.07) .. (498.51,116.06) ;

\draw (132.96,121.89) node [anchor=north west][inner sep=0.75pt]    {$s$};
\draw (506.03,120.88) node [anchor=north west][inner sep=0.75pt]    {$t$};
\draw (360.19,103.9) node [anchor=north west][inner sep=0.75pt]   [align=left] {. . .};
\draw (181.1,153.86) node [anchor=north west][inner sep=0.75pt]    {$H$};
\draw (277.11,152.83) node [anchor=north west][inner sep=0.75pt]    {$H$};
\draw (443.51,152.69) node [anchor=north west][inner sep=0.75pt]    {$H$};
\draw (318.8,213.46) node [anchor=north west][inner sep=0.75pt]    {$N$};
\draw (187.06,121.77) node [anchor=north west][inner sep=0.75pt]    {$0$};
\draw (185.88,64.35) node [anchor=north west][inner sep=0.75pt]    {$1$};
\draw (280.8,64.35) node [anchor=north west][inner sep=0.75pt]    {$1$};
\draw (443.68,65.52) node [anchor=north west][inner sep=0.75pt]    {$1$};
\draw (278.46,121.77) node [anchor=north west][inner sep=0.75pt]    {$0$};
\draw (444.85,122.94) node [anchor=north west][inner sep=0.75pt]    {$0$};
\draw (108.18,50.06) node [anchor=north west][inner sep=0.75pt]  [font=\large]  {$G_{N}^{( 1)}$};

\end{tikzpicture}
    \caption{The construction of $G^{(1)}_N$.}
    \label{fig:first_tightness_proof}
\end{figure}
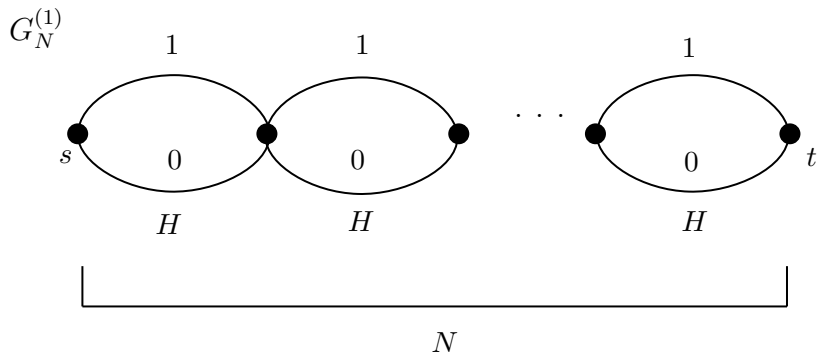
For every $i=1, \dots , N$, let $P_i$ be the $s$-$t$ path in $G^{(1)}_N$ obtained by concatenating the edge of cost zero in every $H$-block with the exception of the $i$-{th} block, from which instead we take the edge of cost 1. Let  $\mathcal{D}$ be the uniform probability distribution over $\{P_i\}_{i=1}^N$. Now, let $\psE$ be the \textit{true} expectation $\mathbb{E}_{\mathcal{D}}$ under this probability distribution. Then $\psE$ is a feasible solution to the SoS relaxation for the $\ell_p$-Shortest Path. Note that since every $P_i$ has $\ell_p$-cost $1$, the $\ell_p^p$-objective value of this SoS solution is also $1$. 

We will now show that the expected cost of the path output by our algorithm when rounding this SoS solution is no smaller than $\bell_1(p) -o (1)$ (where the $o(1)$ term goes to $0$ as $N\to \infty$); this lower bound matches the upper bound from Lemma~\ref{lem:cost-analysis}. \Cref{alg:rounding} will return a path $P$ constructed by selecting one edge at a time from every $H$-block in $G^{(1)}_N$. For each block, the algorithm will pick the edge of cost zero with probability ${(N-1 )/ N}$ and the edge of cost $1$ with probability $1/N$.
Hence, the cost of the path $P$ (since the construction has scalar costs, we the cost of a path here is simply the sum of the scalar costs of edges along the path) follows a binomial distribution with parameters $N$ and ${1/N}$. As $N \to \infty$, the distribution of the cost tends to the Poisson distribution with rate $1$, and so the expectation of the \(p\)-th power of the cost tends to $\bell_1(p)$ by \Cref{fact:iterated-poisson-moments}. Hence, for the case of $h=1$ the bound in Lemma~\ref{lem:cost-analysis} is tight.

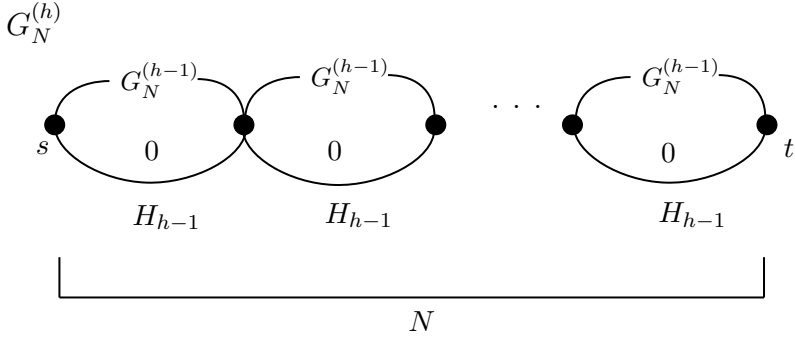
\begin{figure}[t]
    \centering
    \tikzset{every picture/.style={line width=0.75pt}} 

\begin{tikzpicture}[x=0.75pt,y=0.75pt,yscale=-1,xscale=1]

\draw  [fill={rgb, 255:red, 0; green, 0; blue, 0 }  ,fill opacity=1 ] (139.08,115.17) .. controls (139.08,112.56) and (141.2,110.43) .. (143.82,110.43) .. controls (146.44,110.43) and (148.56,112.56) .. (148.56,115.17) .. controls (148.56,117.79) and (146.44,119.92) .. (143.82,119.92) .. controls (141.2,119.92) and (139.08,117.79) .. (139.08,115.17) -- cycle ;
\draw  [fill={rgb, 255:red, 0; green, 0; blue, 0 }  ,fill opacity=1 ] (233.49,115.17) .. controls (233.49,112.56) and (235.62,110.43) .. (238.24,110.43) .. controls (240.86,110.43) and (242.98,112.56) .. (242.98,115.17) .. controls (242.98,117.79) and (240.86,119.92) .. (238.24,119.92) .. controls (235.62,119.92) and (233.49,117.79) .. (233.49,115.17) -- cycle ;
\draw  [fill={rgb, 255:red, 0; green, 0; blue, 0 }  ,fill opacity=1 ] (329.22,115.17) .. controls (329.22,112.56) and (331.35,110.43) .. (333.96,110.43) .. controls (336.58,110.43) and (338.71,112.56) .. (338.71,115.17) .. controls (338.71,117.79) and (336.58,119.92) .. (333.96,119.92) .. controls (331.35,119.92) and (329.22,117.79) .. (329.22,115.17) -- cycle ;
\draw  [fill={rgb, 255:red, 0; green, 0; blue, 0 }  ,fill opacity=1 ] (494.45,115.17) .. controls (494.45,112.56) and (496.58,110.43) .. (499.2,110.43) .. controls (501.82,110.43) and (503.94,112.56) .. (503.94,115.17) .. controls (503.94,117.79) and (501.82,119.92) .. (499.2,119.92) .. controls (496.58,119.92) and (494.45,117.79) .. (494.45,115.17) -- cycle ;
\draw    (143.82,116.06) .. controls (144.53,129.73) and (166.86,143.54) .. (189.81,144.25) .. controls (212.77,144.96) and (239.28,131.79) .. (239.55,116.06) ;
\draw  [fill={rgb, 255:red, 0; green, 0; blue, 0 }  ,fill opacity=1 ] (397.41,115.17) .. controls (397.41,112.56) and (399.54,110.43) .. (402.16,110.43) .. controls (404.78,110.43) and (406.9,112.56) .. (406.9,115.17) .. controls (406.9,117.79) and (404.78,119.92) .. (402.16,119.92) .. controls (399.54,119.92) and (397.41,117.79) .. (397.41,115.17) -- cycle ;
\draw    (146.18,202.02) -- (497.72,202.02) ;
\draw    (146.18,202.02) -- (146.18,181.75) ;
\draw    (497.72,202.02) -- (497.72,181.75) ;
\draw    (143.82,116.06) .. controls (143.23,99.07) and (158.13,93.15) .. (171.13,93.15) ;
\draw    (237.56,117.24) .. controls (238.27,130.9) and (260.6,144.71) .. (283.56,145.42) .. controls (306.51,146.14) and (333.02,132.96) .. (333.29,117.24) ;
\draw    (215.13,92.15) .. controls (231.13,92.27) and (239.23,102.22) .. (238.24,115.17) ;
\draw    (402.78,116.06) .. controls (403.5,129.73) and (425.82,143.54) .. (448.78,144.25) .. controls (471.74,144.96) and (498.25,131.79) .. (498.51,116.06) ;
\draw    (238.82,114.06) .. controls (238.23,97.07) and (253.13,91.15) .. (266.13,91.15) ;
\draw    (310.13,90.15) .. controls (326.13,90.27) and (334.23,100.22) .. (333.24,113.17) ;
\draw    (403.82,114.06) .. controls (403.23,97.07) and (418.13,91.15) .. (431.13,91.15) ;
\draw    (475.13,90.15) .. controls (491.13,90.27) and (499.23,100.22) .. (498.24,113.17) ;

\draw (132.96,121.89) node [anchor=north west][inner sep=0.75pt]    {$s$};
\draw (506.03,120.88) node [anchor=north west][inner sep=0.75pt]    {$t$};
\draw (360.19,103.9) node [anchor=north west][inner sep=0.75pt]   [align=left] {. . .};
\draw (181.1,153.86) node [anchor=north west][inner sep=0.75pt]    {$H_{h-1}$};
\draw (277.11,152.83) node [anchor=north west][inner sep=0.75pt]    {$H_{h-1}$};
\draw (443.51,152.69) node [anchor=north west][inner sep=0.75pt]    {$H_{h-1}$};
\draw (318.8,207.46) node [anchor=north west][inner sep=0.75pt]    {$N$};
\draw (187.06,121.77) node [anchor=north west][inner sep=0.75pt]    {$0$};
\draw (278.46,121.77) node [anchor=north west][inner sep=0.75pt]    {$0$};
\draw (444.85,122.94) node [anchor=north west][inner sep=0.75pt]    {$0$};
\draw (175.18,82.06) node [anchor=north west][inner sep=0.75pt]  [font=\small]  {$G_{N}^{( h-1)}$};
\draw (118.18,52.06) node [anchor=north west][inner sep=0.75pt]  [font=\large]  {$G_{N}^{( h)}$};
\draw (270.18,80.06) node [anchor=north west][inner sep=0.75pt]  [font=\small]  {$G_{N}^{( h-1)}$};
\draw (435.18,80.06) node [anchor=north west][inner sep=0.75pt]  [font=\small]  {$G_{N}^{( h-1)}$};

\end{tikzpicture}
    \caption{The construction of $G^{(h)}_N$}
    \label{fig:second-tightness-proof}
\end{figure}

\paragraph{Generalization to higher \(h\)} We will now lift the construction above to show an analogous result for every value of $h$. We define a family of graphs $\{G^{(h)}_N\}$ recursively as follows.  Define block $H_{h-1}$ as the parallel composition of $G^{(h-1)}_N$ and an edge of cost zero. Then $G^{(h)}_N$ is obtained by combining $N$ distinct $H_{h-1}$-blocks in series (see Figure~\ref{fig:second-tightness-proof}).

In this new construction, we consider a family of $s$-$t$  paths where each path enters a copy of $G_N^{(h-1)}$ exactly in one $H_{h-1}$-block and uses the $0$-cost edges in all other $H_{h-1}$-blocks. Inside this copy of $G_N^{(h-1)}$, it enters a copy of $G_N^{(h-2)}$ in exactly one block $H_{h-2}$-block, and so on.
Note that these paths can be parameterized by the unique edge $e$ of cost $1$ they use. Similarly to the construction above, we consider the SoS solution defined by the expectation $\mathbb{E}_{\mathcal{D}_N^{(h)}}$, where $\mathcal{D}_N^{(h)}$ is the uniform distribution over this family of paths. As before, each path in this family has cost $1$ and so does the SoS solution. 

Now, consider the behavior of our rounding algorithm on the solution \(\mathbb{E}_{\mathcal{D}_N^{(h)}}\).
The distribution of the cost is described by the following recursively-defined sequence (as above, the cost of a path is simply the sum of the scalar costs of the edges along the path).
We let
\[Q_N^{(1)} = 
\begin{cases}1 & \text{w.p. } \frac{1}{N} \\ 0 & \text{w.p. } 1 - \frac{1}{N} \end{cases} \qquad, \qquad R_N^{(1)} = \text{sum of \(N\) independent samples of } Q_N^{(1)}\]
And similarly, for higher \(h \ge 2\) we define
\begin{gather*}
Q_N^{(h)} = \begin{cases} \text{sample from } R_N^{(h-1)} & \text{w.p. } \frac{1}{N} \\ 0 & \text{w.p. } 1 - \frac{1}{N} \end{cases}\\
R_N^{(h)} = \text{sum of \(N\) independent samples from } Q_N^{(h)}
\end{gather*}
Let $P_N^{(h)}$ be the random path in \(G_N^{(1)}\) returned by our algorithm by rounding \(\mathcal{D}_N^{(1)}\). As described above, the distribution of the cost of $P_B^{(1)}$ follows \(R_N^{(1)}\).
We show by induction on $h$ that, in general, the cost of $P = P_N^{(h)}$ follows \(R_N^{(h)}\).
Indeed, consider how path $P$ traverses an $H_{h-1}$-block. It either takes the \(0\)-cost edge or enters a copy of \(G_{N}^{(h-1)}\) inside the block. In the latter case,
the restriction of $P$ to the block is distributed exactly like $P^{(h-1)}_N$; in turn, the cost of $P_N^{(h-1)}$ follows $R_{N}^{(h-1)}$. Since $P$ enters a copy of $G_{N}^{(h-1)}$ with probability $1/N$, the cost of the restriction of $P$ to the $H_{h-1}$-block is distributed like $Q_N^{(h)}$. As the cost of $P$ is the sum of costs of restrictions of $P$ to all $H_{h-1}$-blocks forming $G^{(h)}_N$, we conclude that the cost of $P$ follows $Q_N^{(h)}$, as desired.

Next, we argue that as \(N \to \infty\), the distribution of \(R_N^{(h)}\) approaches \(Z_h\) (the Poisson branching process described in Section~\ref{sec:prelim}).
As above, \(R_N^{(1)}\) converges to \(Z_1\sim \Pois(1)\) in distribution, as $N\to \infty$. Now consider $R_N^{(h)}$. It is a sum of $R^{(1)}_N$ independent copies of $R^{(h-1)}_N$. In the limit when $N\to \infty$, $R_N^{(h)}$ is distributed like a sum of $Z_1$ independent copies of $Z_{h-1}$. It is straightforward that this sum has the same distribution as $Z_h$.

We see that the expected $\ell_p^p$-cost of the path produced by the algorithm is $\E{Z_h^p} = \bell_h(p)$ in the limit (see Fact~\ref{fact:iterated-poisson-moments}); at the same time, the cost of the SoS-solution is 1. Thus, the gap is $\bell_d(p)$. This construction shows that the performance guarantee in Lemma~\ref{lem:cost-analysis} is tight when we run our algorithm on a possibly subotimal SoS solution. 

\paragraph{Tightness of rounding optimal solutions} In the example above, we considered suboptimal SoS solutions; in fact, our instances had optimal solutions of cost 0. Now we consider slightly more complex constructions that show that even when we round optimal SoS solutions, the bound of Lemma~\ref{lem:cost-analysis} can be tight. In contrast to the the constructions above, where costs were scalars, we will use multidimensional vector costs in our new constructions.

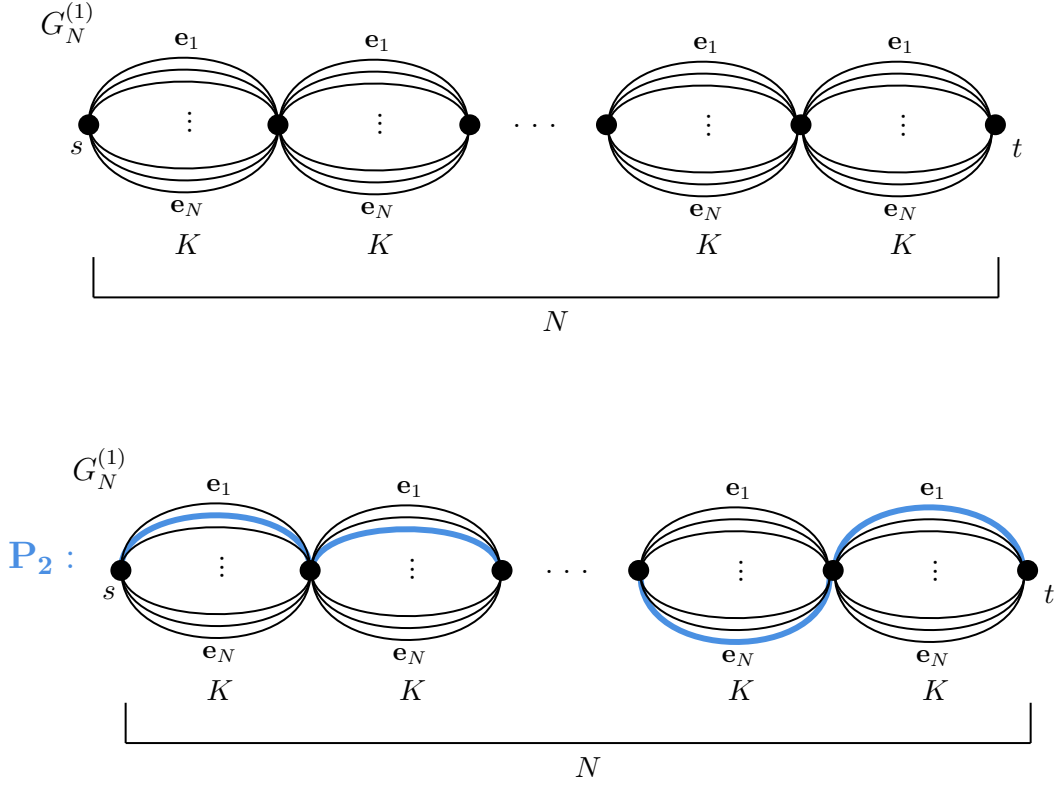
\begin{figure}[t]
    \centering
    \tikzset{every picture/.style={line width=0.75pt}} 

\begin{tikzpicture}[x=0.75pt,y=0.75pt,yscale=-1,xscale=1]

\draw  [fill={rgb, 255:red, 0; green, 0; blue, 0 }  ,fill opacity=1 ] (76.08,92.17) .. controls (76.08,89.56) and (78.2,87.43) .. (80.82,87.43) .. controls (83.44,87.43) and (85.56,89.56) .. (85.56,92.17) .. controls (85.56,94.79) and (83.44,96.92) .. (80.82,96.92) .. controls (78.2,96.92) and (76.08,94.79) .. (76.08,92.17) -- cycle ;
\draw  [fill={rgb, 255:red, 0; green, 0; blue, 0 }  ,fill opacity=1 ] (170.49,92.17) .. controls (170.49,89.56) and (172.62,87.43) .. (175.24,87.43) .. controls (177.86,87.43) and (179.98,89.56) .. (179.98,92.17) .. controls (179.98,94.79) and (177.86,96.92) .. (175.24,96.92) .. controls (172.62,96.92) and (170.49,94.79) .. (170.49,92.17) -- cycle ;
\draw  [fill={rgb, 255:red, 0; green, 0; blue, 0 }  ,fill opacity=1 ] (266.22,92.17) .. controls (266.22,89.56) and (268.35,87.43) .. (270.96,87.43) .. controls (273.58,87.43) and (275.71,89.56) .. (275.71,92.17) .. controls (275.71,94.79) and (273.58,96.92) .. (270.96,96.92) .. controls (268.35,96.92) and (266.22,94.79) .. (266.22,92.17) -- cycle ;
\draw  [fill={rgb, 255:red, 0; green, 0; blue, 0 }  ,fill opacity=1 ] (431.45,92.17) .. controls (431.45,89.56) and (433.58,87.43) .. (436.2,87.43) .. controls (438.82,87.43) and (440.94,89.56) .. (440.94,92.17) .. controls (440.94,94.79) and (438.82,96.92) .. (436.2,96.92) .. controls (433.58,96.92) and (431.45,94.79) .. (431.45,92.17) -- cycle ;
\draw  [fill={rgb, 255:red, 0; green, 0; blue, 0 }  ,fill opacity=1 ] (334.41,92.17) .. controls (334.41,89.56) and (336.54,87.43) .. (339.16,87.43) .. controls (341.78,87.43) and (343.9,89.56) .. (343.9,92.17) .. controls (343.9,94.79) and (341.78,96.92) .. (339.16,96.92) .. controls (336.54,96.92) and (334.41,94.79) .. (334.41,92.17) -- cycle ;
\draw    (83.18,179.02) -- (534.72,179.02) ;
\draw    (83.18,179.02) -- (83.18,158.75) ;
\draw    (534.72,179.02) -- (534.72,158.75) ;
\draw    (80.82,92.17) .. controls (80.23,47.7) and (175.23,46.7) .. (175.24,92.17) ;
\draw    (80.82,92.17) .. controls (80.23,55.7) and (175.23,54.7) .. (175.24,92.17) ;
\draw    (80.82,92.17) .. controls (80.23,63.7) and (175.23,62.7) .. (175.24,92.17) ;
\draw    (175.45,93.07) .. controls (175.45,137.55) and (80.44,137.29) .. (81.04,91.82) ;
\draw    (175.45,93.07) .. controls (175.55,129.55) and (80.55,129.29) .. (81.04,91.82) ;
\draw    (175.45,93.07) .. controls (175.66,121.55) and (80.65,121.29) .. (81.04,91.82) ;
\draw    (175.82,93.17) .. controls (175.23,48.7) and (270.23,47.7) .. (270.24,93.17) ;
\draw    (175.82,93.17) .. controls (175.23,56.7) and (270.23,55.7) .. (270.24,93.17) ;
\draw    (175.82,93.17) .. controls (175.23,64.7) and (270.23,63.7) .. (270.24,93.17) ;
\draw    (270.45,94.07) .. controls (270.45,138.55) and (175.44,138.29) .. (176.04,92.82) ;
\draw    (270.45,94.07) .. controls (270.55,130.55) and (175.55,130.29) .. (176.04,92.82) ;
\draw    (270.45,94.07) .. controls (270.66,122.55) and (175.65,122.29) .. (176.04,92.82) ;
\draw    (339.82,94.17) .. controls (339.23,49.7) and (434.23,48.7) .. (434.24,94.17) ;
\draw    (339.82,94.17) .. controls (339.23,57.7) and (434.23,56.7) .. (434.24,94.17) ;
\draw    (339.82,94.17) .. controls (339.23,65.7) and (434.23,64.7) .. (434.24,94.17) ;
\draw    (434.45,95.07) .. controls (434.45,139.55) and (339.44,139.29) .. (340.04,93.82) ;
\draw    (434.45,95.07) .. controls (434.55,131.55) and (339.55,131.29) .. (340.04,93.82) ;
\draw    (434.45,95.07) .. controls (434.66,123.55) and (339.65,123.29) .. (340.04,93.82) ;
\draw  [fill={rgb, 255:red, 0; green, 0; blue, 0 }  ,fill opacity=1 ] (528.45,92.17) .. controls (528.45,89.56) and (530.58,87.43) .. (533.2,87.43) .. controls (535.82,87.43) and (537.94,89.56) .. (537.94,92.17) .. controls (537.94,94.79) and (535.82,96.92) .. (533.2,96.92) .. controls (530.58,96.92) and (528.45,94.79) .. (528.45,92.17) -- cycle ;
\draw    (436.82,94.17) .. controls (436.23,49.7) and (531.23,48.7) .. (531.24,94.17) ;
\draw    (436.82,94.17) .. controls (436.23,57.7) and (531.23,56.7) .. (531.24,94.17) ;
\draw    (436.82,94.17) .. controls (436.23,65.7) and (531.23,64.7) .. (531.24,94.17) ;
\draw    (531.45,95.07) .. controls (531.45,139.55) and (436.44,139.29) .. (437.04,93.82) ;
\draw    (531.45,95.07) .. controls (531.55,131.55) and (436.55,131.29) .. (437.04,93.82) ;
\draw    (531.45,95.07) .. controls (531.66,123.55) and (436.65,123.29) .. (437.04,93.82) ;

\draw (69.96,98.89) node [anchor=north west][inner sep=0.75pt]    {$s$};
\draw (122.1,145.86) node [anchor=north west][inner sep=0.75pt]    {$K$};
\draw (305.8,184.46) node [anchor=north west][inner sep=0.75pt]    {$N$};
\draw (120.06,128.77) node [anchor=north west][inner sep=0.75pt]  [font=\small]  {$\mathbf{e}_{N}$};
\draw (55.18,29.06) node [anchor=north west][inner sep=0.75pt]  [font=\large]  {$G_{N}^{( 1)}$};
\draw (122.06,44.77) node [anchor=north west][inner sep=0.75pt]  [font=\small]  {$\mathbf{e}_{1}$};
\draw (133.5,82.5) node [anchor=north west][inner sep=0.75pt]  [rotate=-90] [align=left] {{\large ...}};
\draw (218.1,145.86) node [anchor=north west][inner sep=0.75pt]    {$K$};
\draw (382.1,145.86) node [anchor=north west][inner sep=0.75pt]    {$K$};
\draw (215.06,129.77) node [anchor=north west][inner sep=0.75pt]  [font=\small]  {$\mathbf{e}_{N}$};
\draw (217.06,45.77) node [anchor=north west][inner sep=0.75pt]  [font=\small]  {$\mathbf{e}_{1}$};
\draw (228.5,83.5) node [anchor=north west][inner sep=0.75pt]  [rotate=-90] [align=left] {{\large ...}};
\draw (379.06,130.77) node [anchor=north west][inner sep=0.75pt]  [font=\small]  {$\mathbf{e}_{N}$};
\draw (381.06,46.77) node [anchor=north west][inner sep=0.75pt]  [font=\small]  {$\mathbf{e}_{1}$};
\draw (392.5,84.5) node [anchor=north west][inner sep=0.75pt]  [rotate=-90] [align=left] {{\large ...}};
\draw (540.03,97.88) node [anchor=north west][inner sep=0.75pt]    {$t$};
\draw (479.1,145.86) node [anchor=north west][inner sep=0.75pt]    {$K$};
\draw (476.06,130.77) node [anchor=north west][inner sep=0.75pt]  [font=\small]  {$\mathbf{e}_{N}$};
\draw (478.06,46.77) node [anchor=north west][inner sep=0.75pt]  [font=\small]  {$\mathbf{e}_{1}$};
\draw (489.5,84.5) node [anchor=north west][inner sep=0.75pt]  [rotate=-90] [align=left] {{\large ...}};
\draw (291,90.4) node [anchor=north west][inner sep=0.75pt]    {$.\ .\ .$};

\end{tikzpicture}
    \tikzset{every picture/.style={line width=0.75pt}} 

\begin{tikzpicture}[x=0.75pt,y=0.75pt,yscale=-1,xscale=1]

\draw [color={rgb, 255:red, 74; green, 144; blue, 226 }  ,draw opacity=1 ][line width=2.25]    (436.82,94.17) .. controls (436.23,49.7) and (531.23,48.7) .. (531.24,94.17) ;
\draw [color={rgb, 255:red, 74; green, 144; blue, 226 }  ,draw opacity=1 ][line width=2.25]    (434.45,95.07) .. controls (434.45,139.55) and (339.44,139.29) .. (340.04,93.82) ;
\draw [color={rgb, 255:red, 74; green, 144; blue, 226 }  ,draw opacity=1 ][line width=2.25]    (175.82,93.17) .. controls (175.23,64.7) and (270.23,63.7) .. (270.24,93.17) ;
\draw [color={rgb, 255:red, 74; green, 144; blue, 226 }  ,draw opacity=1 ][line width=2.25]    (80.82,92.17) .. controls (80.23,55.7) and (175.23,54.7) .. (175.24,92.17) ;
\draw  [fill={rgb, 255:red, 0; green, 0; blue, 0 }  ,fill opacity=1 ] (76.08,92.17) .. controls (76.08,89.56) and (78.2,87.43) .. (80.82,87.43) .. controls (83.44,87.43) and (85.56,89.56) .. (85.56,92.17) .. controls (85.56,94.79) and (83.44,96.92) .. (80.82,96.92) .. controls (78.2,96.92) and (76.08,94.79) .. (76.08,92.17) -- cycle ;
\draw  [fill={rgb, 255:red, 0; green, 0; blue, 0 }  ,fill opacity=1 ] (170.49,92.17) .. controls (170.49,89.56) and (172.62,87.43) .. (175.24,87.43) .. controls (177.86,87.43) and (179.98,89.56) .. (179.98,92.17) .. controls (179.98,94.79) and (177.86,96.92) .. (175.24,96.92) .. controls (172.62,96.92) and (170.49,94.79) .. (170.49,92.17) -- cycle ;
\draw  [fill={rgb, 255:red, 0; green, 0; blue, 0 }  ,fill opacity=1 ] (266.22,92.17) .. controls (266.22,89.56) and (268.35,87.43) .. (270.96,87.43) .. controls (273.58,87.43) and (275.71,89.56) .. (275.71,92.17) .. controls (275.71,94.79) and (273.58,96.92) .. (270.96,96.92) .. controls (268.35,96.92) and (266.22,94.79) .. (266.22,92.17) -- cycle ;
\draw  [fill={rgb, 255:red, 0; green, 0; blue, 0 }  ,fill opacity=1 ] (431.45,92.17) .. controls (431.45,89.56) and (433.58,87.43) .. (436.2,87.43) .. controls (438.82,87.43) and (440.94,89.56) .. (440.94,92.17) .. controls (440.94,94.79) and (438.82,96.92) .. (436.2,96.92) .. controls (433.58,96.92) and (431.45,94.79) .. (431.45,92.17) -- cycle ;
\draw  [fill={rgb, 255:red, 0; green, 0; blue, 0 }  ,fill opacity=1 ] (334.41,92.17) .. controls (334.41,89.56) and (336.54,87.43) .. (339.16,87.43) .. controls (341.78,87.43) and (343.9,89.56) .. (343.9,92.17) .. controls (343.9,94.79) and (341.78,96.92) .. (339.16,96.92) .. controls (336.54,96.92) and (334.41,94.79) .. (334.41,92.17) -- cycle ;
\draw    (83.18,179.02) -- (534.72,179.02) ;
\draw    (83.18,179.02) -- (83.18,158.75) ;
\draw    (534.72,179.02) -- (534.72,158.75) ;
\draw    (80.82,92.17) .. controls (80.23,47.7) and (175.23,46.7) .. (175.24,92.17) ;
\draw    (80.82,92.17) .. controls (80.23,63.7) and (175.23,62.7) .. (175.24,92.17) ;
\draw    (175.45,93.07) .. controls (175.45,137.55) and (80.44,137.29) .. (81.04,91.82) ;
\draw    (175.45,93.07) .. controls (175.55,129.55) and (80.55,129.29) .. (81.04,91.82) ;
\draw    (175.45,93.07) .. controls (175.66,121.55) and (80.65,121.29) .. (81.04,91.82) ;
\draw    (175.82,93.17) .. controls (175.23,48.7) and (270.23,47.7) .. (270.24,93.17) ;
\draw    (175.82,93.17) .. controls (175.23,56.7) and (270.23,55.7) .. (270.24,93.17) ;
\draw    (270.45,94.07) .. controls (270.45,138.55) and (175.44,138.29) .. (176.04,92.82) ;
\draw    (270.45,94.07) .. controls (270.55,130.55) and (175.55,130.29) .. (176.04,92.82) ;
\draw    (270.45,94.07) .. controls (270.66,122.55) and (175.65,122.29) .. (176.04,92.82) ;
\draw    (339.82,94.17) .. controls (339.23,49.7) and (434.23,48.7) .. (434.24,94.17) ;
\draw    (339.82,94.17) .. controls (339.23,57.7) and (434.23,56.7) .. (434.24,94.17) ;
\draw    (339.82,94.17) .. controls (339.23,65.7) and (434.23,64.7) .. (434.24,94.17) ;
\draw    (434.45,95.07) .. controls (434.55,131.55) and (339.55,131.29) .. (340.04,93.82) ;
\draw    (434.45,95.07) .. controls (434.66,123.55) and (339.65,123.29) .. (340.04,93.82) ;
\draw  [fill={rgb, 255:red, 0; green, 0; blue, 0 }  ,fill opacity=1 ] (528.45,92.17) .. controls (528.45,89.56) and (530.58,87.43) .. (533.2,87.43) .. controls (535.82,87.43) and (537.94,89.56) .. (537.94,92.17) .. controls (537.94,94.79) and (535.82,96.92) .. (533.2,96.92) .. controls (530.58,96.92) and (528.45,94.79) .. (528.45,92.17) -- cycle ;
\draw    (436.82,94.17) .. controls (436.23,57.7) and (531.23,56.7) .. (531.24,94.17) ;
\draw    (436.82,94.17) .. controls (436.23,65.7) and (531.23,64.7) .. (531.24,94.17) ;
\draw    (531.45,95.07) .. controls (531.45,139.55) and (436.44,139.29) .. (437.04,93.82) ;
\draw    (531.45,95.07) .. controls (531.55,131.55) and (436.55,131.29) .. (437.04,93.82) ;
\draw    (531.45,95.07) .. controls (531.66,123.55) and (436.65,123.29) .. (437.04,93.82) ;

\draw (69.96,98.89) node [anchor=north west][inner sep=0.75pt]    {$s$};
\draw (122.1,145.86) node [anchor=north west][inner sep=0.75pt]    {$K$};
\draw (305.8,184.46) node [anchor=north west][inner sep=0.75pt]    {$N$};
\draw (120.06,128.77) node [anchor=north west][inner sep=0.75pt]  [font=\small]  {$\mathbf{e}_{N}$};
\draw (55.18,29.06) node [anchor=north west][inner sep=0.75pt]  [font=\large]  {$G_{N}^{( 1)}$};
\draw (122.06,44.77) node [anchor=north west][inner sep=0.75pt]  [font=\small]  {$\mathbf{e}_{1}$};
\draw (133.5,82.5) node [anchor=north west][inner sep=0.75pt]  [rotate=-90] [align=left] {{\large ...}};
\draw (218.1,145.86) node [anchor=north west][inner sep=0.75pt]    {$K$};
\draw (382.1,145.86) node [anchor=north west][inner sep=0.75pt]    {$K$};
\draw (215.06,129.77) node [anchor=north west][inner sep=0.75pt]  [font=\small]  {$\mathbf{e}_{N}$};
\draw (217.06,45.77) node [anchor=north west][inner sep=0.75pt]  [font=\small]  {$\mathbf{e}_{1}$};
\draw (228.5,83.5) node [anchor=north west][inner sep=0.75pt]  [rotate=-90] [align=left] {{\large ...}};
\draw (379.06,130.77) node [anchor=north west][inner sep=0.75pt]  [font=\small]  {$\mathbf{e}_{N}$};
\draw (381.06,46.77) node [anchor=north west][inner sep=0.75pt]  [font=\small]  {$\mathbf{e}_{1}$};
\draw (392.5,84.5) node [anchor=north west][inner sep=0.75pt]  [rotate=-90] [align=left] {{\large ...}};
\draw (540.03,97.88) node [anchor=north west][inner sep=0.75pt]    {$t$};
\draw (479.1,145.86) node [anchor=north west][inner sep=0.75pt]    {$K$};
\draw (476.06,130.77) node [anchor=north west][inner sep=0.75pt]  [font=\small]  {$\mathbf{e}_{N}$};
\draw (478.06,46.77) node [anchor=north west][inner sep=0.75pt]  [font=\small]  {$\mathbf{e}_{1}$};
\draw (489.5,84.5) node [anchor=north west][inner sep=0.75pt]  [rotate=-90] [align=left] {{\large ...}};
\draw (291,90.4) node [anchor=north west][inner sep=0.75pt]    {$.\ .\ .$};
\draw (23,76.4) node [anchor=north west][inner sep=0.75pt]  [font=\Large]  {$\textcolor[rgb]{0.29,0.56,0.89}{\mathbf{P_{2} :}}$};

\end{tikzpicture}
    \caption{The construction for the discussion below. At the top, we illustrate the structure of the instance graph. At the bottom, we highlight the structure of the path $P_2$ as an example.}
    \label{fig:third-tightness-proof}
\end{figure}

As earlier, we present a recursive construction. We start with $h=1$. Let $K$ be a graph on two vertices with $N$ edges between them, with edge costs $\mathbf{e}_1, \dots ,\mathbf{e}_N$ respectively, where $\mathbf{e}_i$ is the $i$-{th} standard basis vector. Define the graph $G_N^{(1)}$ to be the series composition of $N$ copies of $K$. Observe that the optimal $\ell_p^p$-cost of this instance is $N$: every path that uses exactly one edge of cost $\mathbf{e}_1$, one of cost $\mathbf{e}_2$, \dots, and one of cost $\mathbf{e}_N$ has $\ell_p^p$-cost $N$; every other $s$-$t$ path has $\ell_p^p$-cost strictly greater than $N$. Consider the following family of $\ell_p$-shortest $s$-$t$ paths for this instance. We let $P_i$ be the $s$-$t$ path that takes the $i$-{th} edge from the first copy of $K$ in $G_{N}^{(1)}$, and uses the $i+1$-{st} edge from the second copy of $K$, and so on; we use $(i - 1) + j$ edge from the $j$-{th} copy of $K$ for every $j\in [N]$ (here, we count edges modulo $N$ and identify edges number $i$ and $i+N$). The cost vector of each of these paths equals to the all-one vector in $\mathbb{R}^N$. Let $\cal D$ be the uniform distribution over these paths. Consider the solution to the SoS relaxation defined by the true expectation w.r.t. $\cal D$, the $\ell_p^p$-cost of this solution is $N$. 

We next argue that the SoS solution corresponding to \(\cal D\)
is an optimal solution for the SoS-relaxation described in \Cref{sec:sos-relaxation}.
Consider any other SoS solution to this program, given by \(\psE\).
Let \(f_i \defeq \sum_e c_e(i) x_e\). 
From the SoS relaxation's flow constraints, we get that \(\sum_{i=1}^N \psE\pb{f_i} = N\) (note that $\ell= N$ in our instance).
Using that \((p,0,\dots,0) \succeq (1, \ldots, 1)\) and applying \Cref{lemma:ps_major} and H\"older's inequality, we get the following lower bound on the SoS objective
\begin{align*}
\psE\pb{\sum_{i=1}^N f_i^p} = \sum_{i=1}^N \psE[f_i^p] \ge \sum_{i=1}^N \psE[f_i]^p \geq
N\left(\frac{\sum_{i=1}^N \psE[f_i])}{N}\right)^p = N
\end{align*}
This shows that no other SoS solution has a smaller cost than \(\bE_{\cal D}\).

Now, consider the path $P$ returned by running our rounding algorithm on this SoS solution. Each coordinate of the cost vector associated with $P$ follows a binomial distribution with parameters $N$ and $1/N$ and hence, as $N\to \infty$ the expected $\ell_p^p$-cost of $P$ tends to $N \E{Z_1^p} = N \bell_1(p)$, where $Z_1$ is a Poisson random variable with rate $1$. This shows that the bound is tight for $h=1$.

We will now lift this construction and show that our analysis is tight for arbitrary values of $h$. We define $G^{(h)}_N$ as a graph with vector costs in $\bigl(\R^N\bigr)^{\otimes h} \cong \R^{(N^h)}$. We first compose $N$ copies of $G^{(h-1)}_N$ in parallel to obtain a graph we call an $H_{h-1}$-block. Then we compose $N$ copies of $H_{h-1}$ in series to obtain $G^{(h)}_N$.
We define edge costs for this graph as follows. 
Consider an edge $e$ in an $H_{h-1}$-block such that it belongs to the $i$-th copy of $G_N^{h-1}$ within \(H_{h-1}\), and let \(c_e\) be its cost in $G_N^{(h-1)}$. Then we assign $e$ cost $\mathbf{e}_i\otimes c_e$ within \(G_N^{(h)}\) (we assign the same cost vector to corresponding edges in all copies of $H_{h-1}$).

Alternatively, edge costs can be described as follows. Consider the sequence of indices \(1 \le i_1, \ldots, i_h \le N\)
associated with \(e\): let \(i_1\) be the index of the copy of \(G_N^{(h-1)}\) within \(H_{h-1}\) that contains $e$,  \(i_2\) be the index of the copy of \(G_N^{(h-2)}\) within \(H_{h-2}\) that contains $e$ (here, \(H_{h-2}\) itself is within the \(i_1\)-th copy of \(G_N^{(h-1)}\)), and so on.
Then \(e\) has cost vector \(\mathbf{e}_{i_1} \otimes \cdots \otimes \mathbf{e}_{i_h}\).

Next, we describe the distribution \(\mathcal{D}\) over paths \(G_N^{(h)}\) for which we will analyze the rounding.
Recall that for \(G_N^{(1)}\), we defined \(N\) paths \(P_1, \ldots, P_N\)
such that the costs of edges on each $P_i$ are distinct basis vectors.
Let \((P_i)_{j}\) be the index of the \(j\)-th edge of \(P_i\) within the corresponding $K$-block.
For all indices \(1 \le i_1, \ldots, i_h \le N\), we inductively define path \(P = P_{i_1, \ldots, i_h}\) in $G^{(h)}_N$ as follows.
Path $P$ visits the \((P_{i_1})_1\)-st copy of \(G_N^{(h-1)}\) within the first \(H_{h-1}\)-block in \(G_N^{(h)}\), then the \((P_{i_1})_2\)-th copy of \(G_N^{(h-1)}\) within the second \(H_{h-1}\)-block, and so on up to the \((P_{i_1})_N\)-th copy of \(G_N^{(h-1)}\) within the \(N\)-th \(H_{h-1}\)-block.
Within each copy of \(G_N^{(h-1)}\) visited, \(P\) follows \(P_{i_2, \ldots, i_h}\). Observe that each \(P\) visits each standard basis vector of $\left(\R^{N}\right)^{\otimes h}\) exactly once and thus its total cost vector is $\left(\sum_{i=1}^N \mathbf{e}_i\right)^{\otimes h}$, the all-ones vector in the standard basis.
Thus, the $\ell_p^p$-cost of $P$ is $N^h$.
We then take \(\cal D\) to be uniform over this set of paths $P_{i_1, \ldots, i_h}$.
A similar argument as in the \(h=1\) case then holds to show that \(\cal D\) is an optimal solution for the SoS relaxation. 

Now, consider the output of running the rounding algorithm on \(\bE_{\mathcal{D}}\).
Due to linearity of expectation, we may consider each coordinate of the cost vector independently.
It is easy to see that the distribution of the cost for the \(i\)-th coordinate follows \(R_N^{(h)}\) as defined in the previous construction,
and the expectation of the \(p\)-th power of its cost again satisfies \(\E{Z_d^p} = \bell_h(p)\) in the limit.
We conclude that the expected $\ell_p^p$-cost of the path is $N^h \bell_h(p)$, whereas the SoS objective is $N^h$. This shows that the analysis in \Cref{lem:cost-analysis} is tight for all \(h\).
\newpage
\printbibliography

@article{laddha2022socially,
  title={Socially fair network design via iterative rounding},
  author={Laddha, Aditi and Singh, Mohit and Vempala, Santosh S},
  journal={Operations Research Letters},
  volume={50},
  number={5},
  pages={536--540},
  year={2022},
  publisher={Elsevier}
}

@article{kasperski2009approximability,
  title={On the approximability of minmax (regret) network optimization problems},
  author={Kasperski, Adam and Zieli{\'n}ski, Pawe{\l}},
  journal={Information Processing Letters},
  volume={109},
  number={5},
  pages={262--266},
  year={2009},
  publisher={Elsevier}
}

@incollection{ruzika2009survey,
  title={A survey on multiple objective minimum spanning tree problems},
  author={Ruzika, Stefan and Hamacher, Horst W},
  booktitle={Algorithmics of Large and Complex Networks: Design, Analysis, and Simulation},
  pages={104--116},
  year={2009},
  publisher={Springer}
}

@book{ehrgott2005multicriteria,
  title={Multicriteria optimization},
  author={Ehrgott, Matthias},
  volume={491},
  year={2005},
  publisher={Springer Science \& Business Media}
}

@article{hamacher1994spanning,
  title={On spanning tree problems with multiple objectives},
  author={Hamacher, Horst W and Ruhe, G{\"u}nter},
  journal={Annals of Operations Research},
  volume={52},
  number={4},
  pages={209--230},
  year={1994},
  publisher={Springer}
}

@article{kasperski2016robust,
  title={Robust discrete optimization under discrete and interval uncertainty: A survey},
  author={Kasperski, Adam and Zieli{\'n}ski, Pawe{\l}},
  journal={Robustness analysis in decision aiding, optimization, and analytics},
  pages={113--143},
  year={2016},
  publisher={Springer}
}

@inproceedings{li2023polylogarithmic,
  title={Polylogarithmic Approximation for Robust $s$-$t$ Path},
  author={Li, Shi and Xu, Chenyang and Zhang, Ruilong},
  booktitle={Proceedings of the International Colloquium on Automata, Languages, and Programming (to appear, also available as preprint arXiv:2305.16439)},
  year={2024}
}

@article{breugem2017analysis,
  title={Analysis of FPTASes for the multi-objective shortest path problem},
  author={Breugem, Thomas and Dollevoet, Twan and van den Heuvel, Wilco},
  journal={Computers \& Operations Research},
  volume={78},
  pages={44--58},
  year={2017},
  publisher={Elsevier}
}

@inproceedings{hansen1980bicriterion,
  title={Bicriterion path problems},
  author={Hansen, Pierre},
  booktitle={Multiple Criteria Decision Making Theory and Application: Proceedings of the Third Conference Hagen/K{\"o}nigswinter, West Germany, August 20--24, 1979},
  pages={109--127},
  year={1980},
  organization={Springer}
}

@inproceedings{ravi1993many,
  title={Many birds with one stone: Multi-objective approximation algorithms},
  author={Ravi, Ramamoorthi and Marathe, Madhav V and Ravi, Sekharipuram S and Rosenkrantz, Daniel J and Hunt III, Harry B},
  booktitle={Proceedings of the Symposium on Theory of Computing},
  pages={438--447},
  year={1993}
}

@article{martins1984multicriteria,
  title={On a multicriteria shortest path problem},
  author={Martins, Ernesto Queiros Vieira},
  journal={European Journal of Operational Research},
  volume={16},
  number={2},
  pages={236--245},
  year={1984},
  publisher={Elsevier}
}

@article{kasperski2018approximating,
  title={Approximating some network problems with scenarios},
  author={Kasperski, Adam and Zielinski, Pawel},
  journal={arXiv preprint arXiv:1806.08936},
  year={2018}
}

@inproceedings{chuzhoy2006hardness,
  title={Hardness of directed routing with congestion},
  author={Chuzhoy, Julia and Khanna, Sanjeev},
  booktitle={Electronic Colloquium on Computational Complexity (ECCC)},
  volume={13},
  number={109},
  year={2006}
}

@inproceedings{reich1989beyond,
  title={Beyond Steiner's problem: A VLSI oriented generalization},
  author={Reich, Gabriele and Widmayer, Peter},
  booktitle={International Workshop on Graph-theoretic Concepts in Computer Science},
  pages={196--210},
  year={1989},
  organization={Springer}
}

@article{garg2000polylogarithmic,
  title={A polylogarithmic approximation algorithm for the group Steiner tree problem},
  author={Garg, Naveen and Konjevod, Goran and Ravi, Ramamoorthi},
  journal={Journal of Algorithms},
  volume={37},
  number={1},
  pages={66--84},
  year={2000},
  publisher={Elsevier}
}

@inproceedings{charikar1998rounding,
  title={Rounding via trees: deterministic approximation algorithms for group Steiner trees and k-median},
  author={Charikar, Moses and Chekuri, Chandra and Goel, Ashish and Guha, Sudipto},
  booktitle={Proceedings of the Symposium on Theory of Computing},
  pages={114--123},
  year={1998}
}

@inproceedings{fakcharoenphol2003tight,
  title={A tight bound on approximating arbitrary metrics by tree metrics},
  author={Fakcharoenphol, Jittat and Rao, Satish and Talwar, Kunal},
  booktitle={Proceedings of the Symposium on Theory of Computing},
  pages={448--455},
  year={2003}
}

@article{de2011set,
  title={Set partitions and moments of random variables},
  author={de la Cal, Jes{\'u}s and C{\'a}rcamo, Javier},
  journal={Journal of mathematical analysis and applications},
  volume={378},
  number={1},
  pages={16--22},
  year={2011},
  publisher={Elsevier}
}

@book{stanley2023enumerative,
  title={Enumerative Combinatorics: Volume 2},
  author={Stanley, Richard},
  year={2023},
  publisher={Cambridge University Press}
}

@article{FKP19,
  title={Semialgebraic proofs and efficient algorithm design},
  author={Noah Fleming and Pravesh Kothari and Toniann Pitassi},
  journal={Foundations and Trends in Theoretical Computer Science},
  volume={14},
  number={1-2},
  pages={1--221},
  year={2019},
  issn = {1551-305X},
}

@article{M1902,
  title={Some methods applicable to identities and inequalities of symmetric algebraic functions of $n$ letters},
  author={Muirhead, Robert Franklin},
  journal={Proceedings of the Edinburgh Mathematical Society},
  volume={21},
  pages={144--162},
  year={1902},
  publisher={Cambridge University Press}
}

@article{K1932,
  title={Sur une in{\'e}galit{\'e} relative aux fonctions convexes},
  author={Karamata, Jovan},
  journal={Publications de l'Institut mathematique},
  volume={1},
  number={1},
  pages={145--147},
  year={1932},
  publisher={Matemati{\v{c}}ki institut SANU}
}

@book{micciancio2002complexity,
  title={Complexity of lattice problems: a cryptographic perspective},
  author={Micciancio, Daniele and Goldwasser, Shafi},
  volume={671},
  year={2002},
  publisher={Springer Science \& Business Media}
}

@article{dinur2003approximating,
  title={Approximating CVP to Within Almost-Polynomial Factors is NP-Hard},
  author={Dinur, Irit and Kindler, Guy and Safra, Shmuel},
  journal={Combinatorica},
  volume={23},
  number={2},
  pages={205--243},
  year={2003},
  publisher={Springer-Verlag GmbH}
}

@article{aissi2007approximation,
  title={Approximation of min--max and min--max regret versions of some combinatorial optimization problems},
  author={Aissi, Hassene and Bazgan, Cristina and Vanderpooten, Daniel},
  journal={European Journal of Operational Research},
  volume={179},
  number={2},
  pages={281--290},
  year={2007},
  publisher={Elsevier}
}

@inproceedings{RW17,
  title={On the Bit Complexity of Sum-of-Squares Proofs},
  author={Raghavendra, Prasad and Weitz, Benjamin},
  booktitle={Proceedings of the International Colloquium on Automata, Languages, and Programming},
  volume={80},
  pages={80},
  year={2017}
}

@inproceedings{BKS14,
    author={Boaz Barak and Jonathan Kelner and David Steurer},
    title = {Rounding sum-of-squares relaxations},
    booktitle = {Proceedings of the Symposium on Theory of Computing (also available as preprint arXiv:1312.6652)},
    pages = {31–40},
    title={Rounding Sum-of-Squares Relaxations},      
    year = {2014},

}

@article{charikar1999approximation,
  title={Approximation algorithms for directed Steiner problems},
  author={Charikar, Moses and Chekuri, Chandra and Cheung, To-Yat and Dai, Zuo and Goel, Ashish and Guha, Sudipto and Li, Ming},
  journal={Journal of Algorithms},
  volume={33},
  number={1},
  pages={73--91},
  year={1999},
  publisher={Elsevier}
}

@article{savitch1970,
  title={Relationships between nondeterministic and deterministic tape complexities},
  author={Savitch, Walter J},
  journal={Journal of computer and system sciences},
  volume={4},
  number={2},
  pages={177--192},
  year={1970},
  publisher={Elsevier}
}

@inproceedings{bilo2017simple,
  title={Simple Greedy Algorithms for Fundamental Multidimensional Graph Problems},
  author={Bil{\`o}, Vittorio and Caragiannis, Ioannis and Fanelli, Angelo and Flammini, Michele and Monaco, Gianpiero},
  booktitle={Proceedings of the International Colloquium on Automata, Languages, and Programming},
  year={2017},
  organization={Schloss Dagstuhl-Leibniz-Zentrum fuer Informatik}
}

@inproceedings{chekuri2005recursive,
  title={A Recursive Greedy Algorithm for Walks in Directed Graphs},
  author={Chekuri, Chandra and P\'al, Martin},
  booktitle={Proceedings of the Symposium on Foundations of Computer Science},
  pages={245--253},
  year={2005}
}

@inproceedings{chekuri2010dependent,
  title={Dependent randomized rounding via exchange properties of combinatorial structures},
  author={Chekuri, Chandra and Vondr{\'a}k, Jan and Zenklusen, Rico},
  booktitle={Proceedings of the Symposium on Foundations of Computer Science},
  pages={575--584},
  year={2010},
}

@article{kasperski2011approximability,
  title={On the approximability of robust spanning tree problems},
  author={Kasperski, Adam and Zieli{\'n}ski, Pawe{\l}},
  journal={Theoretical Computer Science},
  volume={412},
  number={4-5},
  pages={365--374},
  year={2011},
  publisher={Elsevier}
}

@article{dinur2002approximating,
  title={Approximating {SVP}$_\infty$ to within almost-polynomial factors is NP-hard},
  author={Dinur, Irit},
  journal={Theoretical Computer Science},
  volume={285},
  number={1},
  pages={55--71},
  year={2002},
  publisher={Elsevier}
}
\newpage

\appendix
\section{Properties of Multidimensional Bell Numbers}\label{sec:missing-proofs}
In this section, we prove a recurrence formula for $d$-dimensional Bell numbers (Lemma~\ref{claim:bell-recurrence}) and an upper bound on Bell numbers stated in \Cref{claim:bellNumberAsymptotics}.

\bellRecurrence*
\begin{proof}
To get a $d$ dimensional partition of $p$, we can first get a labelled partition $P_1$ of $p$ and then get a $(d-1)$-dimensional partition for each set $S\in P_i$. Now for every unlabeled partition $\lambda \prt p$, there are $\binom{p}{\lambda} / \left/ \prod_{j=1}^p \cnt(j, \lambda)!\right.$ ways to choose partition $P_1$ with parts of sizes $\lambda_i$ (in some order). The claim follows.
\end{proof}

\bellNumberAsymptotics*
\begin{proof}
From \eqref{eq:bell:recurrence}, we get $\bell_d(p) \leq \frac{p!}{x^p} (f_d(x) - 1)$ for every $x > 0$. By Stirling's approximation, $(p!)^{1/p} = \Theta(p)$. Thus, $\abell_d(p) = \bell_d(p)^{1/p} \leq O(p)\frac{(f_d(x)-1)^{1/p}}{x}$. We show that for $x = \log(1 + \frac{1}{2d+1})$, $f_i(x)\leq 1/(2d+1 -i)$ for $i\leq d$ by induction on $i$.
For $i=0$, $f_i(x) = f_0(x) = \exp(x) = 1 + \frac{1}{2d+1}$. Using the induction hypothesis and inequality $e^x \leq 1+ x+x^2$ for $x\in[0,1]$, we prove the statement for $i+1$,
$$f_{i+1}(x) = \exp(f_{i}(x) - 1) \leq \exp\left(\frac{1}{2d+1-i}\right) \leq 
\frac{1}{2d+1-i} \cdot \p{1+ \frac{1}{2d+1-i}} \leq \frac{1}{2d-i}.
$$
We get $\abell_d(p) = O(p\cdot d^{-1/p}/x) = O(pd^{1-1/p})$. 

Plugging in $x= \log^{j} p$ into the upper bound $\abell_d(p) \leq O(p) f_d(x)^{1/p}/x$, we get $f_d(x) \leq \exp(p)$ and thus $\abell_d(p) \leq O(p)\frac{O(1)}{\log^{(d)} p}$, as required.
\end{proof}

\section{Discussion on a Dijkstra-style algorithm for \texorpdfstring{$\ell_p$}{l-p} Shortest Path}\label{sec:dijkstra}
In this section, we discuss the Dijkstra-style algorithm for the $\ell_p$-Shortest Path problem given by Bilò, Caragiannis, Fanelli, Flammini, and Monaco~\cite{bilo2017simple}. Theorem 14 in their paper claims that the algorithm yields an $O(\min\{p,\log \ell \})$-approximation. However, we show that Theorem 14 is incorrect, and, in fact, the approximation factor obtained by this algorithm is at least $\Omega(n^{1-1/p})$.

We begin by recalling the algorithm in question, which we refer to as $\ell_p$-Dijkstra:\footnote{We remark that the original version of this algorithm was intended to be applicable in a broader setting and hence had an additional input parameter $E'$. However, when the algorithm is used to compute the $\ell_p$-Shortest Path, $E'=\varnothing$.}

\begin{algorithm}[H]
\caption{$\ell_p$-Dijkstra$(G,d,s,t,p)$~\cite{bilo2017simple}}\label{alg:dijkstra}
\begin{algorithmic}[1]
\For{each $v\in V$}
    \State $\mathsf{PATH}[v] \gets \emptyset$
    \State $\mathsf{DISTANCE}[v] \gets +\infty$
\EndFor
\State $\mathsf{DISTANCE}[s] \gets 0$
\State $Q \gets V$
\While{$Q$ is not empty}
    \State $v \gets \operatorname{argmin}_{u\in Q}\{\mathsf{DISTANCE}[u]\}$
    \For{each $(v,u)\in E$}
    \If{$\mathsf{DISTANCE}[u] > \norm{\boldsymbol{\ell} (\mathsf{PATH}[v]\cup (v,u))}_p$}
    \State $\mathsf{PATH}[u]\gets \mathsf{PATH}[v] \cup (v,u)$
    \State $\mathsf{DISTANCE}[u]\gets\norm{\boldsymbol{\ell} (\mathsf{PATH}[v] 
    \cup (v,u) }_p$
    \EndIf
    \EndFor
    \State $Q \gets Q\setminus \{v\}$
\EndWhile
\State \Return $R$
\end{algorithmic}
\end{algorithm}

Here, for every subset of edges $S$, $\boldsymbol{\ell}(S) = \sum_{e\in S} c_e$.
We now give an example instance illustrating that the result (Theorem 14 in \cite{bilo2017simple}) is false in general. For this example we set $\ell = n + 1$. 

\begin{proposition}\label{prop:unbounded-approx-ration}
    There exists a family $\{G_n\}_{n\in \mathbb{N}}$ of instances of the vector shortest path problem, where for all \(p \ge 1\) the approximation ratio of \Cref{alg:dijkstra} is at least \(\Omega(n^{1-1/p})\).
\end{proposition}

\begin{figure}[H]
    \centering
    \tikzset{every picture/.style={line width=0.75pt}} 

\begin{tikzpicture}[x=0.75pt,y=0.75pt,yscale=-1,xscale=1]

\draw  [fill={rgb, 255:red, 0; green, 0; blue, 0 }  ,fill opacity=1 ] (96.08,112.17) .. controls (96.08,109.56) and (98.2,107.43) .. (100.82,107.43) .. controls (103.44,107.43) and (105.56,109.56) .. (105.56,112.17) .. controls (105.56,114.79) and (103.44,116.92) .. (100.82,116.92) .. controls (98.2,116.92) and (96.08,114.79) .. (96.08,112.17) -- cycle ;
\draw  [fill={rgb, 255:red, 0; green, 0; blue, 0 }  ,fill opacity=1 ] (180.49,112.17) .. controls (180.49,109.56) and (182.62,107.43) .. (185.24,107.43) .. controls (187.86,107.43) and (189.98,109.56) .. (189.98,112.17) .. controls (189.98,114.79) and (187.86,116.92) .. (185.24,116.92) .. controls (182.62,116.92) and (180.49,114.79) .. (180.49,112.17) -- cycle ;
\draw  [fill={rgb, 255:red, 0; green, 0; blue, 0 }  ,fill opacity=1 ] (266.22,112.17) .. controls (266.22,109.56) and (268.35,107.43) .. (270.96,107.43) .. controls (273.58,107.43) and (275.71,109.56) .. (275.71,112.17) .. controls (275.71,114.79) and (273.58,116.92) .. (270.96,116.92) .. controls (268.35,116.92) and (266.22,114.79) .. (266.22,112.17) -- cycle ;
\draw  [fill={rgb, 255:red, 0; green, 0; blue, 0 }  ,fill opacity=1 ] (390.08,111.17) .. controls (390.08,108.56) and (392.2,106.43) .. (394.82,106.43) .. controls (397.44,106.43) and (399.56,108.56) .. (399.56,111.17) .. controls (399.56,113.79) and (397.44,115.92) .. (394.82,115.92) .. controls (392.2,115.92) and (390.08,113.79) .. (390.08,111.17) -- cycle ;
\draw  [fill={rgb, 255:red, 0; green, 0; blue, 0 }  ,fill opacity=1 ] (464.49,111.17) .. controls (464.49,108.56) and (466.62,106.43) .. (469.24,106.43) .. controls (471.86,106.43) and (473.98,108.56) .. (473.98,111.17) .. controls (473.98,113.79) and (471.86,115.92) .. (469.24,115.92) .. controls (466.62,115.92) and (464.49,113.79) .. (464.49,111.17) -- cycle ;
\draw    (100.82,112.17) -- (177.49,112.17) ;
\draw [shift={(180.49,112.17)}, rotate = 180] [fill={rgb, 255:red, 0; green, 0; blue, 0 }  ][line width=0.08]  [draw opacity=0] (8.93,-4.29) -- (0,0) -- (8.93,4.29) -- cycle    ;
\draw    (189.98,112.17) -- (263.22,112.17) ;
\draw [shift={(266.22,112.17)}, rotate = 180] [fill={rgb, 255:red, 0; green, 0; blue, 0 }  ][line width=0.08]  [draw opacity=0] (8.93,-4.29) -- (0,0) -- (8.93,4.29) -- cycle    ;
\draw    (394.82,112.17) -- (461.49,112.17) ;
\draw [shift={(464.49,111.17)}, rotate = 180] [fill={rgb, 255:red, 0; green, 0; blue, 0 }  ][line width=0.08]  [draw opacity=0] (8.93,-4.29) -- (0,0) -- (8.93,4.29) -- cycle    ;
\draw    (100.82,112.17) .. controls (118.22,101.69) and (154.9,102.01) .. (172.38,106.3) ;
\draw [shift={(175.23,107.08)}, rotate = 197.35] [fill={rgb, 255:red, 0; green, 0; blue, 0 }  ][line width=0.08]  [draw opacity=0] (8.93,-4.29) -- (0,0) -- (8.93,4.29) -- cycle    ;
\draw    (100.82,112.17) .. controls (117.97,71.7) and (228.84,61.39) .. (261.78,104.1) ;
\draw [shift={(263.23,106.08)}, rotate = 235.44] [fill={rgb, 255:red, 0; green, 0; blue, 0 }  ][line width=0.08]  [draw opacity=0] (8.93,-4.29) -- (0,0) -- (8.93,4.29) -- cycle    ;
\draw    (100.82,112.17) .. controls (106.18,31.89) and (381.64,29.13) .. (394.51,104.14) ;
\draw [shift={(394.82,106.43)}, rotate = 264.4] [fill={rgb, 255:red, 0; green, 0; blue, 0 }  ][line width=0.08]  [draw opacity=0] (8.93,-4.29) -- (0,0) -- (8.93,4.29) -- cycle    ;
\draw    (100.82,112.17) .. controls (98.26,-7.71) and (431.47,11.68) .. (468.21,103.62) ;
\draw [shift={(469.24,106.43)}, rotate = 251.81] [fill={rgb, 255:red, 0; green, 0; blue, 0 }  ][line width=0.08]  [draw opacity=0] (8.93,-4.29) -- (0,0) -- (8.93,4.29) -- cycle    ;
\draw    (350.08,112.17) -- (387.08,112.17) ; 
\draw [shift={(390.08,112.17)}, rotate = 180.12] [fill={rgb, 255:red, 0; green, 0; blue, 0 }  ][line width=0.08]  [draw opacity=0] (8.93,-4.29) -- (0,0) -- (8.93,4.29) -- cycle    ;
\draw    (275.71,112.17) -- (305.71,112.17) ; 

\draw (93.96,120.89) node [anchor=north west][inner sep=0.75pt]    {$0$};
\draw (180.96,120.89) node [anchor=north west][inner sep=0.75pt]    {$1$};
\draw (263.96,120.89) node [anchor=north west][inner sep=0.75pt]    {$2$};
\draw (313,108) node [anchor=north west][inner sep=0.75pt]   [align=center] {. . .};
\draw (379.96,120.89) node [anchor=north west][inner sep=0.75pt]    {$n-1$};
\draw (462.96,120.89) node [anchor=north west][inner sep=0.75pt]    {$n$};
\draw (131,116.4) node [anchor=north west][inner sep=0.75pt]  [font=\small]  {$\mathbf{e}_{1}$};
\draw (213,116.4) node [anchor=north west][inner sep=0.75pt]  [font=\small]  {$\mathbf{e}_{2}$};
\draw (425,116.4) node [anchor=north west][inner sep=0.75pt]  [font=\small]  {$\mathbf{e}_{n}$};
\draw (127,88.4) node [anchor=north west][inner sep=0.75pt]  [font=\scriptsize]  {$( 1-\varepsilon ) \, \mathbf{e}_{2}$};
\draw (185,62.4) node [anchor=north west][inner sep=0.75pt]  [font=\scriptsize]  {$2( 1-\varepsilon )\, \mathbf{e}_{3}$};
\draw (234,35.4) node [anchor=north west][inner sep=0.75pt]  [font=\scriptsize]  {$( 1-\varepsilon )(n-1)\, \mathbf{e}_{n}$};
\draw (287,15.4) node [anchor=north west][inner sep=0.75pt]  [font=\scriptsize]  {$( 1-\varepsilon ) n\, \mathbf{e}_{n+1}$};

\end{tikzpicture}
    \caption{The counterexample instance for the algorithm of Bilò, Caragiannis, Fanelli, Flammini, and Monaco.}
    \label{fig:dijkstras-counterexample}
\end{figure}
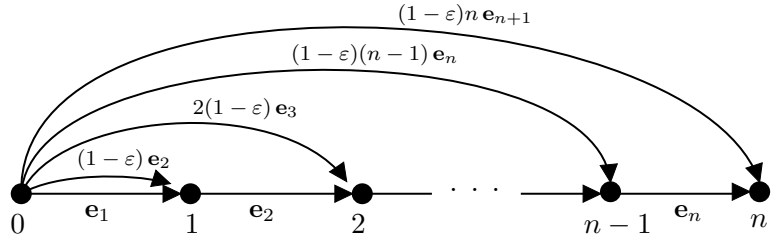
Before proving this proposition, we note that it is in clear contradiction with Theorem 14 in \cite{bilo2017simple}, which claims a constant factor approximation for $p=O(1)$.
\begin{proof}[Proof of Proposition~\ref{prop:unbounded-approx-ration}]
Fix some \(\eps > 0\) which will be chosen later.
We construct the following instance graph $G_n = G=(V,E)$ with $V=\{0, 1, ... , n \}$, $E = E_1 \dot{\cup} E_2$, and terminals $s = 0$ and $t= n$ with (see Figure~\ref{fig:dijkstras-counterexample})
\[
    E_1 := \{ (i-1,i) \mid i\in [n]\},
\]
and
\[
    E_2:= \{ (0,i) \mid i \in [n]\},
\]
Set the vector cost of each edge $(i-1,i) \in E_1$ to the basis vector $\mathbf{e}_i \in \R^{n+1}$, and the cost of each edge $(0,i) \in E_2$ to $(1-\varepsilon)\cdot i\cdot \mathbf{e}_{i+1}\in \R^{n+1}$. 

At the end of the first iteration of the \texttt{while} loop of the algorithm, we have
\[
   \mathsf{DISTANCE}[i] = (1-\varepsilon) i,
\]
and each $\mathsf{PATH}[i]$ (for $i\geq 1$) is simply edge $(0, i)$ from $E_2$. By construction of the weights, the algorithm will never update the values of $\mathsf{DISTANCE}[i]$ and $\mathsf{PATH}[i]$, and thus will return the path from $s=0$ to $t=n$ consisting of the single edge $(0,n)\in E_2$, whose cost is $\norm{(1-\varepsilon) n\, \mathbf{e}_{n+1}}_p = (1-\varepsilon) n$. However, the instance contains the path from $s$ to $t$ formed by the edges of $E_1$; the $\ell_p$ cost of this path is
\[
    \bigg\|\sum_{i\in [n]} \mathbf{e}_i\bigg\|_p  = n^{1/p}.
\]
Hence, the approximation factor of the algorithm on this instance is $(1-\varepsilon)n^{1-1/p}$, where $\varepsilon$ can be chosen arbitrarily small. This concludes the proof.
\end{proof}
\end{document}